%% file: p.tex
\documentclass[a4paper,UKenglish,cleveref, autoref, thm-restate]{lipics-v2021}

\usepackage{amsmath,amssymb,amsfonts}
\usepackage{algorithmic}
\usepackage{graphicx}
\usepackage{textcomp}
\usepackage[dvipsnames]{xcolor}
\usepackage{hyperref}
\usepackage{amsthm}
\usepackage{stmaryrd}
\usepackage{upgreek}
\usepackage{quiver}
\usepackage{adjustbox}

\title{The Groupoid-Syntax of Type Theory is a Set}

\author{Thorsten Altenkirch}{University of Nottingham, United Kingdom}{txa@cs.nott.ac.uk}{https://orcid.org/0000-0002-6582-5025}{}
\author{Ambrus Kaposi}{Eötvös Loránd University (ELTE), Budapest, Hungary}{akaposi@inf.elte.hu}{https://orcid.org/0000-0001-9897-8936}{}
\author{Szumi Xie}{Eötvös Loránd University (ELTE), Budapest, Hungary}{szumi@inf.elte.hu}{https://orcid.org/0009-0001-1355-1114}{}

\authorrunning{T.\ Altenkirch, A.\ Kaposi, and Sz.\ Xie} 

\Copyright{Thorsten Altenkirch, Ambrus Kaposi, and Szumi Xie} 



\category{} 

\relatedversion{} 
\relatedversiondetails{Preprint}{https://arxiv.org/abs/2509.14988}

\supplementdetails[subcategory={Source Code}]{Software}{https://bitbucket.org/akaposi/cohtt}

\funding{The first author was supported by project no. TKP2021-NVA-29
  which has been implemented with the support provided by the Ministry
  of Culture and Innovation of Hungary from the National Research,
  Development and Innovation Fund, ﬁnanced under the TKP2021-NVA
  funding scheme. The second and third authors were funded by the
  European Union (ERC, HOTT, 101170308). Views and opinions expressed
  are however those of the author(s) only and do not necessarily
  reflect those of the European Union or the European Research
  Council. Neither the European Union nor the granting authority can
  be held responsible for them.}

\acknowledgements{We thank the reviewers for their helpful comments and
  suggestions. We thank Niels van der Weide for explaining us the relationship
  to comprehension categories.}

\nolinenumbers 

\EventEditors{Stefano Guerrini and Barbara K\"{o}nig}
\EventNoEds{2}
\EventLongTitle{34th EACSL Annual Conference on Computer Science Logic (CSL 2026)}
\EventShortTitle{CSL 2026}
\EventAcronym{CSL}
\EventYear{2026}
\EventDate{February 23--28, 2026}
\EventLocation{Paris, France}
\EventLogo{}
\SeriesVolume{363}
\ArticleNo{19}

\theoremstyle{plain} 
\newtheorem{thm}{Theorem}
\newtheorem{lemm}[thm]{Lemma}
\newtheorem{prop}[thm]{Proposition}
\newtheorem{problem}[thm]{Problem}

\theoremstyle{definition}
\newtheorem{defn}[thm]{Definition}
\newtheorem{examp}[thm]{Example}
\newtheorem{construction}[thm]{Construction}
\newtheorem{notation}[thm]{Notation}
\newtheorem{parameter}[thm]{Parameter}

\input{abbrevs.tex}
\newcommand{\Agda}{%
  \raisebox{-.05em}{\includegraphics{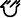}}%
}

\begin{document}

\maketitle

\begin{abstract}
Categories with families (CwFs) have been used to define the semantics of
type theory in type theory. In the setting of Homotopy Type
Theory (HoTT), one of the limitations of the traditional notion of
CwFs is the requirement to set-truncate types, which excludes models
based on univalent categories, such as the standard set model. To
address this limitation, we introduce the concept of a
\textit{Groupoid Category with Families (GCwF)}. This framework
truncates types at the groupoid level and incorporates coherence
equations, providing a natural extension of the CwF framework when
starting from a 1-category.

We demonstrate that the initial GCwF for a type theory with a base
family of sets and $\Pi$-types (groupoid-syntax) is
set-truncated. Consequently, this allows us to utilize the
conventional intrinsic syntax of type theory while enabling
interpretations in semantically richer and more natural models.
All constructions in this paper were formalised in Cubical Agda.
\end{abstract}

\section{Introduction}

In \cite{DBLP:conf/popl/AltenkirchK16}, an \emph{intrinsically typed
  syntax} for basic type theory using a
\emph{Quotient-Inductive-Inductive Type (QIIT)} was introduced. By intrinsically typed, we mean that the syntax directly enforces typing constraints, eliminating the need for separate untyped preterms needed in extrinsic presentations. The equational theory is integrated naturally using \emph{path constructors} from Homotopy Type Theory (HoTT), while set-truncation ensures that types behave as sets.

QIITs are a special case of \emph{Higher Inductive-Inductive Types (HIITs)} where
all types are truncated to sets by adding a higher path
constructor. The term \emph{inductive-inductive} signals that
constructors can reference other constructors in their types. In
essence,  \cite{DBLP:conf/popl/AltenkirchK16}  defined the syntax of type theory as the \emph{initial\footnote{
Initiality is equivalent to induction for any QIIT \cite{DBLP:journals/pacmpl/KaposiKA19}, this been generalised to HIITs by Sattler \cite{christian}. In this paper, we use the terms \emph{initial model} and \emph{syntax} interchangeably.}
Category with Families (CwF)} with $\Pi$-types and an uninterpreted base family. This allowed the syntax to be interpreted in any CwF with the necessary structure and served as the foundation for a proof of normalisation using \emph{Normalisation by Evaluation (NbE)} \cite{altenkirch2017normalisation}.

However, this approach had a significant limitation: the syntax could not be interpreted in the \emph{intended model} where types are sets. This issue arose due to the use of set truncation, which enforced types to be sets but precluded a univalent semantics, such as $\Set$. To work around this, inductive-recursive universes were used. While effective, this approach was unsatisfactory as it excluded univalent models, which are natural semantics for type theory.

Simply omitting set truncation is not a solution. Without truncation: (i) we cannot prove necessary equations in the syntax;
(ii) the syntax itself is no longer a set, which e.g.\ makes equality in the syntax undecidable.
A fully principled solution would require adding all higher coherences. However, this is both technically complex and generally believed to require a \emph{2-level type theory} rather than plain HoTT \cite{DBLP:conf/lics/Kraus21}.


In this paper, we propose a middle ground: we lift the truncation level to \emph{groupoids} and add a minimal set of coherence equations.
This enables interpreting the syntax into the set model and other univalent category-based models.
This compromise aligns naturally with the structure of categories in HoTT \cite{HoTTbook}, where
\emph{objects} are groupoids with no truncation restriction, while
\emph{hom-sets} remain sets, as their name implies.
Actually, if we only restrict types to be groupoids, then we can prove
that contexts in the syntax also form a groupoid.

At first glance, this raises a new concern: does lifting to groupoids and adding coherence equations require redefining the syntax? Do we lose decidability of equality? Our \emph{main result} resolves this concern:

\begin{quote}
\textbf{The groupoid-syntax of type theory with $\Pi$-types and a base family has types and contexts that are sets.}
\end{quote}

In essence, we retain the \emph{set-truncated syntax of type theory}
while enabling evaluation in \emph{groupoid-level models}. This allows
us to interpret the set-truncated syntax into univalent models, such
as $\Set$ or the \emph{container model}
\cite{altenkirch2021container}. However, we note that univalence for
types cannot be assumed as a principle at the judgmental level—doing
so would mean that types are not a set anymore.

\subparagraph*{Contributions.}
The main contributions of this paper are as follows:
\begin{itemize}
    \item We introduce the notion of a \emph{Groupoid Category with Families (GCwF)} with $\Pi$-types and a base family (Definition \ref{def:groupidCwF}).
    \item We show that the initial GCwF with $\Pi$-types and a base family is \emph{set-truncated}.
    \item We establish the above proof using an \emph{$\alpha$-normalisation} construction.
    \item As a result, we enable the definition of the univalent \emph{set model} and other univalent category-based models for the set-truncated syntax.
\end{itemize}
All results are formalised in \emph{Cubical Agda}.

\subparagraph*{Structure of the paper.}
After listing related work, we describe our metatheory and notation in
Section \ref{sec:meta}. In Section \ref{sec:syntax}, we define
various syntaxes as HIITs and describe the problem of interpreting the
set-truncated syntax in sets. In Section \ref{sec:alpha} we show that
the groupoid-syntax is a set. We use this fact in Section
\ref{sec:fruits} to fix the above problem. We conclude in Section
\ref{sec:conclusion}.

\subparagraph*{Related work.}
This paper is a continuation of the series of papers internalising the
intrinsic syntax of type theory in type theory
\cite{DBLP:conf/types/Danielsson06,DBLP:journals/entcs/Chapman09} and
in homotopy type theory
\cite{mike,DBLP:conf/popl/AltenkirchK16}. Intrinsic syntax means that
there are only well-formed, well-scoped, well-typed terms which are
quotiented by conversion. This is in contrast with extrinsic style
formalisations
\cite{DBLP:journals/pacmpl/0001OV18,DBLP:conf/cpp/AdjedjLMPP24}. We
use a variant of Dybjer's CwFs \cite{DBLP:conf/types/Dybjer95}
introduced by Ehrhard \cite{ehrhard,coquandEhrhard}.

Infinite-dimensional versions of our 1-dimensional notion of model are
given by Kraus and Uemura. Kraus defines a notion of $\infty$-CwF
\cite{DBLP:conf/lics/Kraus21} inside an extension of type theory with
a strict equality (two-level type theory,
\cite{DBLP:conf/csl/AltenkirchCK16,DBLP:journals/mscs/AnnenkovCKS23,DBLP:journals/mscs/AnnenkovCKS24}). He
conjectures that the set-level (0-level) syntax is initial for his
$\infty$-model. Uemura \cite{DBLP:journals/corr/abs-2212-11764} proves
normalisation for an $\infty$-dimensional presentation of type theory,
however his work is not formalised in intensional type theory.

Our theorem that the initial GCwF with certain type formers is
set-truncated can be seen as a simple coherence theorem analogous to
that of monoidal categories. Coherence for monoidal categories says
that in the free monoidal category over a set of objects, morphisms
form a set. Our coherence theorem is for types rather than morphisms
(substitutions), and we generate the types from a set-valued family
using $\Pi$ and instantiation. Coherence for monoidal groupoids was
proven in HoTT by Piceghello \cite{piceghello:LIPIcs.TYPES.2019.8},
where he also used groupoid-truncated HITs to define the free monoidal
groupoid.

In HoTT, the ideal 
solution for coherence problems is to find finite descriptions which imply
all the infinitely many coherences. For example, usability of integers defined as set-quotients is limited, but there is a way to define
their $\infty$-version without truncation
\cite{DBLP:conf/lics/AltenkirchS20}. Free groups can defined
without truncation \cite{DBLP:conf/lics/KrausA18}, however originally groupoid-truncation was needed
to prove that the free group over a set is a set. The general case was resolved by Wärn \cite{wärn2024pathspacespushouts}.
The Symmetry book \cite{Symmetry} contains several similar examples. 

There are notions of model of type theory weaker than CwFs where
e.g.\ substitution is only functorial up to isomorphism
\cite{DBLP:conf/csl/Hofmann94,DBLP:journals/tocl/LumsdaineW15,DBLP:conf/types/Bocquet21}.
These are further away from implementations of type theory, but they admit the
standard model in the setting of HoTT \cite{DBLP:conf/lics/Weide25}, even
without going 2-categorical. 2-categories are used to formulate the equivalence
of weaker and stricter notions of model.
Dybjer and Clairambault
\cite{DBLP:journals/mscs/ClairambaultD14} proved the 2-equivalence
of locally cartesian closed categories and CwFs with appropriate
structure. Van der Weide \cite{DBLP:conf/lics/Weide25} formalised this
result in a univalent setting, for univalent comprehension
categories instead of CwFs, and extended it to new type formers.

Higher inductive-inductive types (HIITs) have been used before to
describe free algebraic structures such as real numbers
\cite{HoTTbook}, the partiality monad
\cite{DBLP:conf/fossacs/AltenkirchDK17}, or hybrid semantics
\cite{DBLP:conf/fscd/Diezel020}, but all of these are 
set-truncated HIITs, unlike our groupoid-syntax. Cubical type theory
supports HITs
\cite{DBLP:conf/lics/CoquandHM18,DBLP:journals/pacmpl/CavalloH19}, and
there is a scheme for describing HIITs
\cite{DBLP:journals/lmcs/KaposiK19} which covers our usages.

\section{Metatheory and formalisation}\label{sec:meta}

Everything in this paper is formalised in Cubical Agda
\cite{DBLP:journals/jfp/VezzosiMA21}, the formalisation is available
as supplementary material. Next to definitions/theorems/etc.,
\href{https://csl26-cohtt.github.io/TT.README.html}{\Agda} icons
point to the corresponding part in the HTML version of the source
code. In the paper text, we use informal cubical type theory: this
means that we do not refer to the interval and instead of using
3-dimensional cubes, we only compose and fill larger 2-dimensional
shapes.

We write $\equiv$ for equations holding definitionally, $:\equiv$
denotes definitions. Dependent function space is written $(x:A)\to B$
or $\forall x\ldotp B$, we also use implicit quantifications. We write dependent pairs as $(x:A)\times B$,
the empty type as $\bot$, the unit type as $\top$. The universe of
types is $\Type$, we also use the universe of h-sets $\Set$ and
h-groupoids $\mathsf{Groupoid}$. We have a predicative universe hierarchy, but we do not write levels for readability. The identity (path, equality) type is
written $a =_A b$ for $a, b : A$, where the subscript $_A$ is usually
ommitted. The dependent path type (PathP, heterogeneous equality) is
written $b_0 =_B^e b_1$ for $e : a_0 = a_1$ and $b_i : B\,a_i$,
sometimes the subscript $_B$ is ommitted. We overload functions and their congruence
(ap operator), e.g.\ $f\,e : f\,a = f\,b$ where $e : a = b$, and we
omit symmetries as well. Transport is written $e_*\,b_0 : B\,a_1$ for
$e : a_0 = a_1$ and $b_0 : B\,a_0$, we tend to give a separate name
for operations using transport (e.g.\ $\blank[\blank]^\U$ is a
transported version of $\blank[\blank]$). The obvious element of the heterogeneous equality
$b_0 =_B^e\,e_*\,b_0$ is called $\mathrm{transportFiller}$. The
composition operator of cubical type theory is the generalisation of
transitivity as depicted below, it also comes with a filler operation.
\[\begin{tikzcd}
	& {a_1} & \cdots & {a_{n-1}} \\
	{a_0} &&&& {a_n}
	\arrow["{e_2}", from=1-2, to=1-3]
	\arrow["{e_{n-1}}", from=1-3, to=1-4]
	\arrow["{e_n}", curve={height=-12pt}, from=1-4, to=2-5]
	\arrow["{e_1}", curve={height=-12pt}, from=2-1, to=1-2]
	\arrow[""{name=0, anchor=center, inner sep=0}, "{\text{composition}}"', dotted, from=2-1, to=2-5]
	\arrow["{\text{filler}}"{description}, draw=none, from=1-3, to=0]
\end{tikzcd}\]
For the composite equality $e$, we denote the filler by $\mathrm{fillerOf}\,e$. Some of these composition and filling
operations are primitive in cubical type theory, but they are also
definable via the eliminator of the identity type (J). In this paper
we abstract over these differences.

We write compositions with equational reasoning by $a_0
\overset{e_1}{=} a_1 \overset{e_2}{=} \dots \overset{e_n}{=} a_n$, or
its multi-line variant (left, below). Composition also works for
heterogeneous equalities, in this case we write the base equalities in
superscripts (right, below).
\begin{alignat*}{10}
  & a_0 && {=}(e_1)     &&  b_0 && {=}^{e_1}(e'_1) \\      
  & \dots               && && \dots \\                       
  & a_{n-1} && {=}(e_n) \hspace{9em} &&  b_{n-1} && {=}^{e_n}(e'_n) \\  
  & a_n                && &&   b_n                            
\end{alignat*}
Here $b_i : B\,a_i$, $e_i : a_{i-1} = a_i$ and $e'_i : b_{i-1}
=_B^{e_i} b_i$, and the resulting heterogeneous equality is $b_0
=_B^{\text{composite of the $e_i$s}} b_n$. We denote
heterogeneous composition and filling of shapes by drawing a base
diagram and a dependent diagram. We say that the right diagram is
\emph{over} the left one: in this case the dotted composition line has type $b_0 =_B^{e_0}\,b_n$.
\adjustbox{minipage={1.1\textwidth},scale=0.9,center}{
\[
\begin{tikzcd}[ampersand replacement=\&]
	\& {a_1} \& \cdots \& {a_{n-1}} \&\&\& {b_1} \& \cdots \& {b_{n-1}} \\
	{a_0} \&\&\&\& {a_n} \& {b_0} \&\&\&\& {b_n}
	\arrow["{e_2}", from=1-2, to=1-3]
	\arrow["{e_{n-1}}", from=1-3, to=1-4]
	\arrow["{e_n}", curve={height=-12pt}, from=1-4, to=2-5]
	\arrow["{e'_2}", from=1-7, to=1-8]
	\arrow["{e'_{n-1}}", from=1-8, to=1-9]
	\arrow["{e'_n}", curve={height=-12pt}, from=1-9, to=2-10]
	\arrow["{e_1}", curve={height=-12pt}, from=2-1, to=1-2]
	\arrow[""{name=0, anchor=center, inner sep=0}, "{e_0}"', from=2-1, to=2-5]
	\arrow["{e'_1}", curve={height=-12pt}, from=2-6, to=1-7]
	\arrow[""{name=1, anchor=center, inner sep=0}, "{\text{composition}}"', dotted, from=2-6, to=2-10]
	\arrow["s"{description}, draw=none, from=1-3, to=0]
	\arrow["{\text{filler}}"{description}, draw=none, from=1-8, to=1]
\end{tikzcd}
\]
}
Some squares can be filled by the fact that every parameterised path
is natural. We denote these naturality squares by writing
$\mathrm{nat}$ in the center:
\adjustbox{minipage={\textwidth},scale=0.9,left}{
  \[
\begin{tikzcd}[ampersand replacement=\&]
	{f\,x} \& {f\,y} \\
	{g\,x} \& {g\,y}
	\arrow["{f\,e'}", from=1-1, to=1-2]
	\arrow["{e\,x}"', from=1-1, to=2-1]
	\arrow["{\mathrm{nat}}"{description}, draw=none, from=1-1, to=2-2]
	\arrow["{e\,y}", from=1-2, to=2-2]
	\arrow["{g\,e'}"', from=2-1, to=2-2]
\end{tikzcd}
\]
}

There are some technical limitations of Cubical Agda that we have to
circumvent in the formalisation, but are not visible in the text of
this paper. We summerise these below.
\begin{itemize}
\item Interleaved constructors of (higher) inductive-inductive
  datatypes are not allowed in Cubical Agda. For example, this
  fragment of a syntax of a type theory is not allowed:
  \begin{alignat*}{10}
    & \Con && : \Type                    && \blank\ext\blank && : (\Gamma:\Con)\to\Ty\,\Gamma\to\Con \\         
    & \Ty && : \Con\to\Type \hspace{6em} && \Sigma && : (A:\Ty\,\Gamma)\to\Ty\,(\Gamma\ext A)\to\Ty\,\Gamma \\  
    & &&                                 && \mathsf{eq} && : \Gamma\ext\Sigma\,A\,B =_\Con \Gamma\ext A\ext B   
  \end{alignat*}
  Here every later constructor depends on all the previous
  constructors, the order cannot be modified, and first we have a
  $\Con$-constructor, then a $\Ty$-constructor, then another $\Con$-constructor.
  We solve this via the encoding proposed in \cite{email}, which uses
  the same idea as encoding mutual inductive types as an indexed
  family \cite{DBLP:conf/fscd/KaposiR20}: we introduce a sort of codes
  $\mathsf{Code}$ and a family of elements $\mathsf{EL}$, and then
  all constructors are in the same sort:
  \begin{alignat*}{10}
    & \mathsf{Code} && : \Type                   && \blank\ext\blank && : (\Gamma:\mathsf{EL}\,\Con)\to\mathsf{EL}\,(\Ty\,\Gamma)\to\mathsf{EL}\,\Con \\                    
    & \mathsf{EL} && : \mathsf{Code}\to\Type     && \Sigma && : (A:\mathsf{EL}\,(\Ty\,\Gamma))\to\mathsf{EL}\,(\Ty\,(\Gamma\ext A))\to\mathsf{EL}\,(\Ty\,\Gamma) \\         
    & \Con && : \mathsf{Code}                    && \mathsf{eq} && : \Gamma\ext\Sigma\,A\,B =_{\mathsf{EL}\,\Con} \Gamma\ext A\ext B \\
    & \Ty && : \mathsf{EL}\,\Con\to\mathsf{Code} \hspace{3.9em} &&
  \end{alignat*}
  We use the same technique when defining our syntaxes (Definitions
  \ref{def:wild}, \ref{def:set}, \ref{def:groupoid}).
\item When we describe HIITs, we use transport and composition, but in
  the formalisation, we avoid them (we still use composition operators
  in some 2-paths). The reason is twofold: (i) Agda does not see that
  these operations preserve strict positivity; (ii) as the $\beta$
  rule for transport is not definitional, it makes it difficult to
  formalise strict models such as the $\Type$-interpretation. Instead,
  we make sure that all transports appear outermost and then can be
  encoded via dependent paths (a dependent path on $\refl$ computes to
  a nondependent one). When it is not possible to do this, we add
  extra constructors together with equations which singleton contract
  them. For example, in the text of the paper we write the
  substitution law for $\El$ using a transport:
  $(\El\,\hat{A})[\gamma] =
  \El\,\big((\U[]\,\gamma)_*\,(\hat{A}[\gamma])\big)$. In the formalisation,
  we introduce a new constructor $\blank[\blank]^\U :
  \Tm\,\Gamma\,\U\to\Sub\,\Delta\,\Gamma\to\Tm\,\Delta\,\U$ together
  with the contracting equation $\hat{A}[\gamma] =^{\U[]\,\gamma}
  \hat{A}[\gamma]^\U$, and then use this new constructor when describing
  $\El[]$.
\item When characterising the equality of normal types, in the
  formalisation we use the inductively defined Martin-Löf identity
  type instead of the built-in path type (note that they are
  equivalent). This is convenient because J computes definitionally on
  its constructor $\refl$. In the text of the paper we abstract over this.
\end{itemize}

\section{Variants of the syntax and the set interpretation}\label{sec:syntax}

In this section we define three different variants of the syntax of a
type theory with dependent function space and a base family: the wild
syntax, the set-truncated and the groupoid-truncated syntax. We show
that types in the wild syntax do not form a set, so in particular they
cannot have decidable equality. The set-syntax cannot be interpreted
into the set model directly, while the groupoid-syntax can.

\begin{parameter}
  Everything in this section is parameterised by an $X : \Set$ and a $Y : X\to\Set$.
\end{parameter}

\begin{defn}[Wild syntax \href{https://csl26-cohtt.github.io/TT.Wild.Syntax.html}{\Agda}]\label{def:wild}
  We define a higher inductive-inductive type with four sorts. It
  starts with a category with a terminal object. Objects are called
  contexts and morphisms are called substitutions, the terminal object
  is called the empty context. Note that composition
  $\blank\circ\blank$ takes the $\Gamma$, $\Delta$ and $\Theta$
  arguments implicitly, and similarly for all the forthcoming operations
  and equations.
  \begin{alignat*}{10}
    & \Con && : \Type                                                                                                       && \mathsf{idl} && : \forall\gamma\ldotp\id\circ\gamma = \gamma \\          
    & \Sub && : \Con\to\Con\to\Type                                                                                         && \mathsf{idr} && : \forall\gamma\ldotp\gamma\circ\id = \gamma \\          
    & \blank\circ\blank && : \Sub\,\Delta\,\Gamma\to\Sub\,\Theta\,\Delta\to\Sub\,\Theta\,\Gamma \hspace{5em}                && {\diamond} && : \Con \\                                                  
    & \mathsf{ass} && : \forall\gamma\,\delta\,\theta\ldotp \gamma\circ(\delta\circ\theta) = (\gamma\circ\delta)\circ\theta && \epsilon && : \Sub\,\Gamma\,\diamond \\                                  
    & \id && : \Sub\,\Gamma\,\Gamma                                                                                         && {\diamond}\eta && : (\sigma:\Sub\,\Gamma\,{\diamond})\to\sigma = \epsilon
  \end{alignat*}
  Types form a presheaf over the category of contexts and
  substitutions. The action on morphisms is called instantiation, it
  uses a flipped notation because of contravariance.
  \begin{alignat*}{10}
    & \Ty && : \Con\to\Type                                                               && [{\circ}] && : \forall A\,\gamma\,\delta \ldotp A[\gamma\circ\delta] = A[\gamma][\delta] \\
    & \blank[\blank] && : \Ty\,\Gamma\to\Sub\,\Delta\,\Gamma\to\Ty\,\Delta \hspace{8.4em} && [\id] && : \forall A\ldotp A[\id] = A                                                      
  \end{alignat*}
  Terms form a dependent presheaf over types. The instantiation
  operation is overloaded. Note that the functor laws are paths
  dependent over the functor laws for $\Ty$.
  \begin{alignat*}{10}
    & \Tm && : (\Gamma:\Con)\to\Ty\,\Gamma\to\Type                                                                 && [{\circ}] && : \forall a\,\gamma\,\delta\ldotp a[\gamma\circ\delta] =_{\Tm\,\Theta}^{[{\circ}]\,A\,\gamma\,\delta} a[\gamma][\delta] \\
    & \blank[\blank] && : \Tm\,\Gamma\,A\to(\gamma:\Sub\,\Delta\,\Gamma)\to\Tm\,\Delta\,(A[\gamma]) \hspace{1.6em} && [\id] && : \forall a\ldotp a[\id] =_{\Tm\,\Gamma}^{[\id]\,A} a                                                                         
  \end{alignat*}
  In addition to context extension (infix triangle), we have lifting
  of substitutions which is its functorial action on morphisms. The
  functor laws again depend on those for $\Ty$.
  \begin{alignat*}{10}
    & \blank\ext\blank && : (\Gamma:\Con)\to\Ty\,\Gamma\to\Con                                                &&  {\circ}^+ && : \forall\gamma\,\delta\,\ldotp (\gamma\circ\delta)^+ =^{[{\circ}]\,A\,\gamma\,\delta} \gamma^+ \circ \delta^+ \\
    & \blank^+ && : (\gamma:\Sub\,\Delta\,\Gamma)\to\Sub\,(\Delta\ext A[\gamma])\,(\Gamma\ext A) \hspace{1.9em} &&  \id^+ && : \id^+ =^{[\id]\,A} \id                                                                                             
  \end{alignat*}
  We have weakening $\p$ (or first projection), and zero De Bruijn
  index $\q$ (second projection). We explain how to compose either
  with lifted substitutions.
  \begin{alignat*}{10}
    & \p && : \Sub\,(\Gamma\ext A)\,\Gamma \hspace{1.25em}
    & \q && : \Tm\,(\Gamma\ext A)\,(A[\p]) \hspace{1.25em}
    & \p{\circ}^+ && : \forall\gamma\ldotp \p\circ \gamma^+ = \gamma\circ\p \hspace{1.25em}
    & \q[^+] && : \forall\gamma\ldotp\q[\gamma^+] =^e \q
  \end{alignat*}
  The last equation is heterogeneous over the previous one, $e$
  abbreviates the following composite path in $\Ty\,(\Delta\ext A[\gamma])$:
  $
  A[\p][\gamma^+] \overset{[{\circ}]\,A\,\p\,\gamma^+}{=} A[\p\circ\gamma^+] \overset{\p{\circ}^+\,\gamma}{=} A[\gamma\circ\p] \overset{[{\circ}]\,A\,\gamma\,\p}{=} A[\gamma][\p].
  $

  
  So far we have all weakenings and variables, for example De Bruijn
  index $3$ is given by $\q[\p][\p][\p]$. Now we introduce single
  substitutions via $\langle a\rangle$ which instantiates the last
  variable in the context by $a$, and leaves the rest. It commutes
  with any substitution, and we explain how to compose $\p$ and $\q$
  with single substitutions.
  \begin{alignat*}{10}
    & \langle\blank\rangle && : \Tm\,\Gamma\,A\to\Sub\,\Gamma\,(\Gamma\ext A)                                                             && \p{\circ}\langle\rangle && : \forall a\ldotp\p\circ\langle a\rangle = \id \\
    & \langle\rangle{\circ} && : \forall a\,\gamma\ldotp \langle a\rangle\circ\gamma = \gamma^+\circ\langle a[\gamma]\rangle \hspace{6em} && \q[\langle\rangle] && : \forall a\ldotp \q[\langle a\rangle] =^e a 
  \end{alignat*}
  Again, the last equation is heterogeneous over the previous one,
  where $e$ abbreviates the following path in $\Ty\,\Gamma$:
  $
  A[\p][\langle a\rangle] \overset{[{\circ}]\,A\,\p\,\langle a\rangle}{=} A[\p\circ\langle a\rangle] \overset{\p{\circ}\langle\rangle}{=} A[\id] \overset{[\id]\,A}{=} A.
  $

  
  The last equation for the substitution calculus is an $\eta$ law
  explaining that an identity substitution on an extended context is
  given by $\p$ and $\q$.
  \begin{alignat*}{10}
    & {\ext}\eta && : \id = \p^+ \circ \langle\q\rangle
  \end{alignat*}
  We have a base type and a family over it, and elements of
  these coming from the parameters.
  \begin{alignat*}{10}
    & \U && : \Ty\,\Gamma \hspace{1.9em}
    & \El && : \Tm\,\Gamma\,\U\to\Ty\,\Gamma \hspace{1.9em}
    & \inU && : X\to\Tm\,\diamond\,\U \hspace{1.9em}
    & \inEl && : Y\,x \to \Tm\,\diamond\,(\El\,(\inU\,x))
  \end{alignat*}
  The substitution law for $\U$ is easy. To express $\El[]$, we
  introduce notation for the instantiation operation of terms of type
  $\U$, which is just a transported version of ordinary instantiation.
  \begin{alignat*}{10}
    & \U[] && : \forall\gamma\ldotp\U[\gamma] = \U                                                  && \blank[\blank]^\U && : \Tm\,\Gamma\,\U\to\Sub\,\Delta\,\Gamma\to\Tm\,\Delta\,\U \\
    & \El[] && : \forall\gamma\ldotp(\El\,\hat{A})[\gamma] = \El\,(\hat{A}[\gamma]^\U) \hspace{7em} && \hat{A}[\gamma]^\U && :\equiv (\U[]\,\gamma)_*\,\hat{A}[\gamma]
  \end{alignat*}
  We introduce a transport-filler heterogeneous equality for each
  $\hat{A}$ and $\gamma$ that we will make use of later:
  $
  \hat{A}[\gamma]^\U\mathsf{filler} : \hat{A}[\gamma] =^{\U[]\,\gamma} \hat{A}[\gamma]^\U.
  $
  
  Dependent function space with $\beta$, $\eta$ laws is defined by the
  isomorphism $\Tm\,(\Gamma\ext A)\,B\cong\Tm\,\Gamma\,(\Pi\,A\,B)$
  natural in $\Gamma$. It is enough to state naturality in one
  direction.
  \begin{alignat*}{10}
    & \Pi && : (A:\Ty\,\Gamma)\to\Ty\,(\Gamma\ext A)\to\Ty\,\Gamma                                             &&  \Pi\beta && : \forall b\ldotp\app\,(\lam\,b) = b \\                                             
    & \Pi[] && : \forall A\,B\,\gamma\ldotp (\Pi\,A\,B)[\gamma] = \Pi\,(A[\gamma])\,(B[\gamma^+]) \hspace{1.15em} &&  \Pi\eta && : \forall f\ldotp\lam\,(\app\,f) = f \\                                              
    & \lam && : \Tm\,(\Gamma\ext A)\,B\to\Tm\,\Gamma\,(\Pi\,A\,B)                                             &&   \lam[] && : \forall b\,\gamma\ldotp(\lam\,b)[\gamma] =^{\Pi[]\,A\,B\,\gamma} \lam\,(b[\gamma^+]) \\
    & \app && : \Tm\,\Gamma\,(\Pi\,A\,B)\to\Tm\,(\Gamma\ext A)\,B                                             &&
  \end{alignat*}
  This concludes the definition of the wild syntax.
\end{defn}

We defined the substitution calculus in Ehrhard's style
\cite{ehrhard,coquandEhrhard} instead of the more usual category with
families (CwF) \cite{DBLP:conf/types/Dybjer95,Castellan2021}. These
two presentations of the substitution calculus are isomorphic. In the
above syntax, substitution extension $\blank,\blank :
(\gamma:\Sub\,\Delta\,\Gamma)\to\Tm\,\Delta\,(A[\gamma])\to\Sub\,\Delta\,(\Gamma\ext
A)$ is defined as $(\gamma,a) :\equiv \gamma^+\circ\langle
a\rangle$. In the other direction, $\gamma^+ :\equiv
\big(\gamma\circ\p,([{\circ}]\,A\,\gamma\,\p)_*\,\q\big)$ and
$\langle a\rangle :\equiv (\id,([\id]\,A)_*\,a)$.

Although CwFs have one less operation and fewer equations, we chose the Ehrhard
style syntax as there is no need to use
the transport operation when specifying the equations. In CwFs, the
naturality of substitution extension needs a transport in the middle:
$
(\gamma,a)\circ\delta = \big(\gamma\circ\delta, ([{\circ}]\,A\,\gamma\,\delta)_*\,(a[\delta])\big)
$
In our syntax, all the transports are outermost, hence can be encoded
by dependent paths.

\begin{examp}[Using the wild syntax \href{https://csl26-cohtt.github.io/TT.Wild.Examples.html}{\Agda}]\label{ex:using}
  We derive the other direction of naturality for the
  $\Pi$-isomorphism: this is the substitution law for $\app$ called
  $\app[]$.
  \begin{alignat*}{10}
    & (\app\,t)[\gamma^+] && {=}(\Pi\beta\,t) \\
    & \app\,\Big(\lam\,\big((\app\,t)[\gamma^+]\big)\Big) && {=}(\lam[]\,(\app\,t)\,\gamma^+) \\
    & \app\,\bigg((\Pi[]\,A\,B\,\gamma)_*\,\Big(\big(\lam\,(\app\,t)\big)[\gamma]\Big)\bigg)\; && {=}(\Pi\eta\,t) \\
    & \app\,\big((\Pi[]\,A\,B\,\gamma)_*\,(t[\gamma])\big)
  \end{alignat*}
  Nondependent function space is encoded as $A\Rightarrow B :\equiv
  \Pi\,A\,(B[\p])$.
  
  The identity function for the family $\U$, $\El$ is defined as
  \begin{alignat*}{10}
    & \mathsf{ID} : \Tm\,\diamond\,(\Pi\,\U\,(\El\,((\U[]\,\p)_*\,\q)\Rightarrow\El\,((\U[]\,\p)_*\,\q))) \hspace{5.6em}
    & \mathsf{ID} :\equiv \lam\,(\lam\,\q)
  \end{alignat*}
  Note that we had to transport the zero De Bruijn index $\q :
  \Tm\,(\diamond\ext\U)\,(\U[\p])$ so that we can apply $\El$ to it:
  $(\U[]\,\p)_*\,\q : \Tm\,(\diamond\ext\U)\,\U$.

  In the syntax, we have the categorical application operation for
  $\Pi$. Ordinary application is given by 
  $\blank\cdot\blank : \Tm\,\Gamma\,(\Pi\,A\,B)\to(a:\Tm\,\Gamma\,A)\to\Tm\,\Gamma\,(B[\langle a\rangle])$ defined as
  $t \cdot a :\equiv (\app\,t)[\langle a\rangle]$.
  It is easy to prove its $\beta$ law
  $
  (\lam\,t)\cdot a \equiv \app\,(\lam\,t)[\langle a\rangle] \overset{\Pi\beta\,t}{=} t[\langle a\rangle],
  $
  but the $\eta$ law is more involved as it needs several
  transports. We prove it via heterogeneous equality reasoning, where
  the proof of the equality of the types is written in the superscript
  of the equality sign.
  \begin{alignat*}{10}
    & f && {=}(\Pi\eta\,f) \\
    & \lam\,(\app\,f) && {=^{[\id]\,B}}([\id]\,(\app\,f)) \\
    & \lam\,\big((\app\,f)[\id]\big) && {=^{{\ext}\eta}}({\ext}\eta) \\
    & \lam\,\big((\app\,f)[\p^+\circ\langle\q\rangle]\big) && {=^{[\circ]\,B\,\p^+\,\langle\q\rangle}}([{\circ}]\,(\app\,f)\,\p^+\,\langle\q\rangle) \\
    & \lam\,\big((\app\,f)[\p^+][\langle\q\rangle]\big) && {=}(\app[]\,t\,\p) \\
    & \lam\,\Big(\app\,\big((\Pi[]\,A\,B\,p)_*\,(f[\p])\big)[\langle\q\rangle]\Big)\, && {\equiv} \\
    & \lam\,\big((\Pi[]\,A\,B\,p)_*\,(f[\p])\cdot\q\big)\hspace{1.1em}
  \end{alignat*}
  The type of the above heterogeneous equality is
  $
    f =^e_{\Tm\,\Gamma\,(\Pi\,A\,\blank)} \lam\,\big((\Pi[]\,A\,B\,p)_*\,(f[\p])\cdot\q\big),
  $
  where $e$ is the following composite of the three heterogeneous
  steps in the above equality reasoning:
  $
  B \overset{[\id]\,B}{=} B[\id] \overset{{\ext}\eta}{=} B[\p^+\circ\langle\q\rangle] \overset{[\circ]\,B\,\p^+\,\langle\q\rangle}{=} B[\p^+][\langle\q\rangle].
  $
\end{examp}

\begin{problem}[Type interpretation of the wild syntax \href{https://csl26-cohtt.github.io/TT.Wild.TypeInterp.html}{\Agda}]\label{prob:type_interpretation}
  As a sanity check for our wild syntax, we define its type
  (standard, metacircular) interpretation.\footnote{
  Following Voevodsky \cite{voevodsky2015csystemdefineduniversecategory}, we call a proof relevant theorem a problem and its proof a construction.}
\end{problem}
\begin{proof}[Construction]
We define the following four recursive-recursive functions by pattern
matching on the constructors of the higher inductive-inductive type.
\begin{alignat*}{10}
  & \llbracket\blank\rrbracket && : \Con\to\Type                                                                                && \llbracket\blank\rrbracket && : \Ty\,\Gamma\to\llbracket\Gamma\rrbracket\to\Type                                      \\
  & \llbracket\blank\rrbracket && : \Sub\,\Delta\,\Gamma\to\llbracket\Delta\rrbracket\to\llbracket\Gamma\rrbracket \hspace{6em} && \llbracket\blank\rrbracket && : \Tm\,\Gamma\,A\to(\gamma:\llbracket\Gamma\rrbracket)\to\llbracket A\rrbracket\,\gamma
\end{alignat*}
Composition is function composition
($\llbracket\gamma\circ\delta\rrbracket\,\bar{\theta} :\equiv
\llbracket\gamma\rrbracket\,(\llbracket\delta\rrbracket\,\bar{\theta})$),
identity is identity ($\llbracket\id\rrbracket\,\bar{\gamma} :\equiv
\bar{\gamma}$), instantiation is composition ($\llbracket
A[\gamma]\rrbracket\,\bar{\delta} :\equiv \llbracket
A\rrbracket\,(\llbracket\gamma\rrbracket\,\bar{\delta})$), context
extension is dependent sum ($\llbracket\Gamma\ext A\rrbracket :\equiv
(\bar{\gamma}:\llbracket\Gamma\rrbracket)\times \llbracket
A\rrbracket\,\bar{\gamma}$), lifting is
$\llbracket\gamma^+\rrbracket\,(\bar{\delta},\bar{a}) :\equiv
(\llbracket\gamma\rrbracket\,\bar{\delta},\bar{a})$, $\p$ and $\q$ are
first and second projections, single substitution is
$\llbracket\langle a\rangle\rrbracket\,\bar{\gamma} :\equiv
(\bar{\gamma},\llbracket a\rrbracket\,\bar{\gamma})$. Function space
is interpreted by metatheoretic functions
($\llbracket\Pi\,A\,B\rrbracket\,\bar{\gamma}
:\equiv(\bar{a}:\llbracket A\rrbracket)\to\llbracket
B\rrbracket\,(\bar{\gamma},\bar{a})$). $\U$ and $\El$ are interpreted
by $X$ and $Y$, $\inU$ and $\inEl$ simply return their arguments. All
the equations are $\refl$.
\end{proof}
The standard interpretation shows that our theory is consistent, that
is, not all types are inhabited: $\Tm\,\diamond\,\U$ is interpreted
by $\top\to X$ so it is inhabited if and only if $X$ is.

\begin{prop}[\href{https://csl26-cohtt.github.io/TT.Wild.NotSet.html\#\%C2\%ACisSetTy}{\Agda}]
  Types in the wild syntax ($\Ty\,\Gamma$ for a $\Gamma:\Con$) do not form a set.
\end{prop}
\begin{proof}
  Every higher inductive type, including our Definition \ref{def:wild}
  can be interpreted into the unit type where all paths are
  interpreted by $\refl$. We use a variant of this where every sort is
  interpreted by $\top$ except $\Ty\,\Gamma$ is interpreted by the
  circle $\mathsf{S}^1$. $\Pi$, $\U$ and $\El$ are constant
  $\mathsf{base}$, $A[\gamma]$ is interpreted by the interpretation of
  $A$. All equations are interpreted by $\refl$, except $\U[]$ which
  is interpreted by $\mathsf{loop}$. The two different proofs of
  $\U[\id] = \U$, namely $[\id]\,\U$ and $\U[]\,\id$ are interpreted
  by $\refl$ and $\mathsf{loop}$, respectively.
\end{proof}
When using the wild syntax, this is a practical problem: it
can happen that we need a term of type $\El\,((\U[]\,\id)_*\,a)$, but
we only have a term of type $\El\,(([\id]\,\U)_*\,a)$ available. From
a broader perspective, Hedberg's theorem \cite[Theorem
  7.2.5]{HoTTbook} implies that we cannot prove normalisation for the
wild syntax.
In principle, there could be a clever way of defining the equations in
the syntax such that there is only one proof for each equation. It is
not known whether this is possible \cite{mike}. Instead,
we make all the equations equal by force.
\begin{defn}[Set-syntax \href{https://csl26-cohtt.github.io/TT.Set.Syntax.html}{\Agda}]\label{def:set}
  The set-based syntax is the wild syntax (Definition \ref{def:wild})
  extended with the following three higher equality constructors. They
  truncate substitutions, types and terms to sets.\vspace{-1.5em}
  \begin{alignat*}{10}
    & && && \mathsf{isSetTy} && : (e\,e' : A_0 =_{\Ty\,\Gamma}A_1)\to e = e' \\ 
    & \mathsf{isSetSub} && : (e\,e' : \gamma_0 =_{\Sub\,\Delta\,\Gamma}\gamma_1)\to e = e' \hspace{3em} && \mathsf{isSetTm} && : (e\,e' : a_0 =_{\Tm\,\Gamma\,A}a_1)\to e = e' 
  \end{alignat*}
\end{defn}
We do not add that contexts form a set as it is provable by
induction on the context (\href{https://csl26-cohtt.github.io/TT.Set.ConPath.html\#isSetCon}{\Agda}).

Now we can hope for normalisation for this syntax, but the standard
interpretation does not work anymore: the interpretation of
$\Ty\,\Gamma$ would be $\llbracket\Gamma\rrbracket\to\Type$, but then
we cannot interpret $\mathsf{isSetTy}$, as $\Type$ does not form a
set. We have to limit ourselves to interpreting $\Ty\,\Gamma$ by
$\llbracket\Gamma\rrbracket\to\mathsf{Prop}$ where $\mathsf{Prop}$ is
defined as $(A:\Type)\times\big((x\,y:A)\to x =
y\big)$. Alternatively, we can interpret $\Ty$ into an
inductive-recursive universe as in \cite[Section
6]{DBLP:conf/popl/AltenkirchK16}, 
but we cannot interpret the set-syntax in a univalent model.
To fix this, we introduce a syntax where substitutions and terms are truncated to
be sets, but types are only groupoid-truncated. To make types
well-behaved, we add coherence laws which are equations between
equations between types. These express that the substitution laws
$\U[]$, $\El[]$ and $\Pi[]$ commute with the functoriality laws
$[\circ]$, $[\id]$. In the diagrams below, the vertical directions are
the substitution laws and the horizontal directions are the
functoriality laws.
\begin{defn}[Groupoid-syntax \href{https://csl26-cohtt.github.io/TT.Groupoid.Syntax.html}{\Agda}]\label{def:groupoid}
  The groupoid-based syntax is the wild syntax (Definition
  \ref{def:wild}) extended with the following higher equality
  constructors. Some of them are drawn as commutative diagrams.\vspace{-1.5em}
  \begin{alignat*}{10}
    & && && \mathsf{isGrpdTy} && : (w\,w' : e =_{A_0 =_{\Ty\,\Gamma} A_1} e')\to w = w' \\
    & \mathsf{isSetSub} && : (e\,e' : \gamma_0 =_{\Sub\,\Delta\,\Gamma}\gamma_1)\to e = e' \hspace{1.6em} && \mathsf{isSetTm} && : (e\,e' : a_0 =_{\Tm\,\Gamma\,A}a_1)\to e = e' \\
    & \U[\id] && : [\id]\,\U = \U[]\,\id &&
    \El[\id] && : \forall\hat{A}\ldotp
    \begin{tikzcd}[ampersand replacement=\&]
	{(\El\,\hat{A})[\id]} \\
	{\El\,(\hat{A}[\id]^\U)} \& {\El\,\hat{A}}
	\arrow["{\El[]\,\hat{A}\,\id}"', from=1-1, to=2-1]
	\arrow["{[\id]\,(\El\,\hat{A})}", from=1-1, to=2-2]
	\arrow["{[\id]^\U\,\hat{A}}"', from=2-1, to=2-2]
    \end{tikzcd}
  \end{alignat*}
  \vspace{-3em}
  \begin{alignat*}{10}
    & \U[{\circ}] : \forall\gamma\,\delta\ldotp && \El[{\circ}] : \forall\hat{A}\,\gamma\,\delta\ldotp \\
    &
    \begin{tikzcd}[ampersand replacement=\&]
      {\U[\gamma\circ\delta]} \& {\U[\gamma][\delta]} \\
      \& {\U[\delta]} \\
      \& \U
      \arrow["{[{\circ}]\,\U\,\gamma\,\delta}", from=1-1, to=1-2]
      \arrow["{\U[]\,(\gamma\circ\delta)}"', from=1-1, to=3-2]
      \arrow["{\U[]\,\gamma}", from=1-2, to=2-2]
      \arrow["{\U[]\,\delta}", from=2-2, to=3-2]
    \end{tikzcd} \hspace{8.8em}
    &&
    \begin{tikzcd}[ampersand replacement=\&]
	{(\El\,\hat{A})[\gamma\circ\delta]} \& {(\El\,\hat{A})[\gamma][\delta]} \\
	\& {(\El\,(\hat{A}[\gamma]^\U))[\delta]} \\
	{\El\,(\hat{A}[\gamma\circ\delta]^\U)} \& {\El\,(\hat{A}[\gamma]^\U[\delta]^\U)}
	\arrow["{{[\circ]\,(\El\,\hat{A})\,\gamma\,\delta}}", from=1-1, to=1-2]
	\arrow["{{\El[]\,\hat{A}\,(\gamma\circ\delta)}}"', from=1-1, to=3-1]
	\arrow["{{\El[]\,\hat{A}\,\gamma}}", from=1-2, to=2-2]
	\arrow["{{\El[]\,(\hat{A}[\gamma]^\U)\,\delta}}", from=2-2, to=3-2]
	\arrow["{{[\circ]^\U\,\hat{A}\,\gamma\,\delta}}", from=3-1, to=3-2]
    \end{tikzcd}
  \end{alignat*}\vspace{-2em}
  \begin{alignat*}{10}
    & \Pi[\circ] : \forall\,A\,B\,\gamma\,\delta\ldotp && \Pi[\id] : \forall A\,B\ldotp \\
    & \vspace{-1em}
    \begin{tikzcd}[ampersand replacement=\&]
	{(\Pi\,A\,B)[\gamma\circ\delta]} \& {(\Pi\,A\,B)[\gamma][\delta]} \\
	\& {(\Pi\,(A[\gamma])\,(B[\gamma^+]))[\delta]} \\
	{\Pi\,(A[\gamma\circ\delta])\,(B[(\gamma\circ\delta)^+])} \& {\Pi\,(A[\gamma][\delta])\,(B[\gamma^+][\delta^+])}
	\arrow["{{[\circ]\,(\Pi\,A\,B)\,\gamma\,\delta}}", from=1-1, to=1-2]
	\arrow["{{\Pi[]\,A\,B\,(\gamma\circ\delta)}}"', from=1-1, to=3-1]
	\arrow["{{\Pi[]\,A\,B\,\gamma}}"', from=1-2, to=2-2]
	\arrow["{{\Pi[]\,(A[\gamma])\,(B[\gamma^+])\,\delta}}"', from=2-2, to=3-2]
	\arrow["{\raisebox{-1em}{$\Pi\,({[\circ]\,A\,\gamma\,\delta})\,({[\circ^+]\,B\,\gamma\,\delta})$}}"', from=3-1, to=3-2]
    \end{tikzcd}\hspace{-0.2em}
    &&
\begin{tikzcd}[ampersand replacement=\&]
	{(\Pi\,A\,B)[\id]} \\
	{\Pi\,(A[\id])\,(B[\id^+])} \& {\Pi\,A\,B} \\
	{\phantom{\bullet}} \& {\phantom{\bullet}}
	\arrow["{\Pi[]\,A\,B\,\id}"', from=1-1, to=2-1]
	\arrow["{[\id]\,(\Pi\,A\,B)}", from=1-1, to=2-2]
	\arrow["{\raisebox{-1em}{$\Pi\,([\id]\,A)\,([\id^+]\,B)$}}"', from=2-1, to=2-2]
	\arrow["{\phantom{\Pi}}"', draw=none, from=3-1, to=3-2]
\end{tikzcd}
  \end{alignat*}
  In the types of $\U[\circ]$ and $\U[\id]$ above, $[{\circ}]^\U$ and
  $[\id]^\U$ abbreviate the following equality proofs. $[{\circ}]^\U$
  is the dotted line in the left dependent square which is over the
  right square. $[\id]^\U$ is the dotted line in the upper dependent
  triangle which is over the lower triangle. As the bottom lines in
  the base square/triangle are reflexivities, $[{\circ}]^\U$ and
  $[\id]^\U$ are homogeneous equalities, but all the other lines in
  the upper shapes are heterogeneous. Fillers of the base shapes are
  written in their center, they are operations of the groupoid-syntax
  defined before. In Cubical Agda, the dotted lines are defined via
  heterogeneous composition. The $\blank[\blank]^\U\mathsf{filler}$
  operation is part of Definition \ref{def:wild}.\vspace{-1.5em}
  \[
    \begin{tikzcd}[ampersand replacement=\&]
	{\hat{A}[\gamma\circ\delta]} \& {\hat{A}[\gamma][\delta]} \& {\U[\gamma\circ\delta]} \& {\U[\gamma][\delta]} \\
	\& {\hat{A}[\gamma]^\U[\delta]} \&\& {\U[\delta]} \\
	{\hat{A}[\gamma\circ\delta]^\U} \& {\hat{A}[\gamma]^\U[\delta]^\U} \& \U \& \U
	\arrow["{[{\circ}]\,\hat{A}\,\gamma\,\delta}", from=1-1, to=1-2]
	\arrow["{\hat{A}[\gamma\circ\delta]^\U\mathsf{filler}}"', from=1-1, to=3-1]
	\arrow["{\hat{A}[\gamma]^\U\mathsf{filler}}"', from=1-2, to=2-2]
	\arrow["{[{\circ}]\,\U\,\gamma\,\delta}", from=1-3, to=1-4]
	\arrow["{\U[]\,(\gamma\circ\delta)}"', from=1-3, to=3-3]
	\arrow["{\U[{\circ}]\,\gamma\,\delta}"{description}, draw=none, from=1-3, to=3-4]
	\arrow["{\U[]\,\gamma}", from=1-4, to=2-4]
	\arrow["{\hat{A}[\gamma]^\U[\delta]^\U\mathsf{filler}}"', from=2-2, to=3-2]
	\arrow["{\U[]\,\delta}", from=2-4, to=3-4]
	\arrow["{[{\circ}]^\U\,\hat{A}\,\gamma\,\delta}"', dotted, from=3-1, to=3-2]
	\arrow[equals, from=3-3, to=3-4]
    \end{tikzcd}\hspace{3em}
    \begin{tikzcd}[ampersand replacement=\&]
	{\hat{A}[\id]} \\
	{\hat{A}[\id]^\U} \& {\hat{A}} \\
	{\U[\id]} \\
	\U \& \U
	\arrow["{A[\id]^\U\mathsf{filler}}"', from=1-1, to=2-1]
	\arrow["{[\id]\,\hat{A}}", from=1-1, to=2-2]
	\arrow["{[\id]^\U\,\hat{A}}"', dotted, from=2-1, to=2-2]
	\arrow["{\U[]\,\id}"', from=3-1, to=4-1]
	\arrow[""{name=0, anchor=center, inner sep=0}, "{[\id]\,\U}", from=3-1, to=4-2]
	\arrow[equals, from=4-1, to=4-2]
	\arrow["{\U[\id]}"{description}, draw=none, from=4-1, to=0]
    \end{tikzcd}
  \]
  In the types of $\Pi[{\circ}]$ and $\Pi[\id]$ above, we used the
  following abbreviations of paths. $[{\circ}^+]$ and $[\id^+]$ are
  the dotted lines in the upper triangles, which are over the lower
  triangles. The dotted lines are defined by composition. We also give
  names to the fillers of the upper triangles which will be used in
  Figures \ref{fig:picomp} and \ref{fig:piid}, respectively: \\
\adjustbox{minipage={\textwidth},scale=0.9,left}{
  \[
\begin{tikzcd}[ampersand replacement=\&]
	{B[(\gamma\circ\delta)^+]} \& {B[\gamma^+\circ\delta^+]} \& {A[\gamma\circ\delta]} \& {A[\gamma][\delta]} \\
	\& {B[\gamma^+][\delta^+]} \&\& {A[\gamma][\delta]}
	\arrow["{{\circ}^+\,\gamma\,\delta}", from=1-1, to=1-2]
	\arrow[""{name=0, anchor=center, inner sep=0}, "{[{\circ}^+]\,B\,\gamma\,\delta}"', curve={height=18pt}, dotted, from=1-1, to=2-2]
	\arrow["{[\circ]\,B\,\gamma^+\,\delta^+}", from=1-2, to=2-2]
	\arrow["{[{\circ}]\,A\,\gamma\,\delta}", from=1-3, to=1-4]
	\arrow["{[{\circ}]\,A\,\gamma\,\delta}"', curve={height=18pt}, from=1-3, to=2-4]
	\arrow[equals, from=1-4, to=2-4]
	\arrow["{[{\circ}^+]\mathsf{filler}\,B\,\gamma\,\delta}"{description}, draw=none, from=1-2, to=0]
\end{tikzcd}
\begin{tikzcd}[ampersand replacement=\&]
	{B[\id^+]} \& {B[\id]} \& {A[\id]} \& A \\
	\& B \&\& A
	\arrow["{{\id}^+}", from=1-1, to=1-2]
	\arrow[""{name=0, anchor=center, inner sep=0}, "{[\id^+]\,B}"', curve={height=18pt}, dotted, from=1-1, to=2-2]
	\arrow["{[\id]\,B}", from=1-2, to=2-2]
	\arrow["{[{\id}]\,A}", from=1-3, to=1-4]
	\arrow["{[\id]\,A}"', curve={height=18pt}, from=1-3, to=2-4]
	\arrow[equals, from=1-4, to=2-4]
	\arrow["{[\id^+]\mathsf{filler}\,B}"{description}, draw=none, from=1-2, to=0]
\end{tikzcd}
\]
}

\vspace{1em}

\noindent This concludes the definition of the groupoid-syntax.
\end{defn}
\begin{notation}
We denote the components of the set-syntax by $_\S$ and the
groupoid-syntax by $_\G$ subscripts, e.g.\ $\Con_\S$ and $\Con_\G$.
\end{notation}
We cannot redo the interpretation of Problem
\ref{prob:type_interpretation} because $\Type$ is not a groupoid, but
we can refine it by interpreting types into $\Set$.
\begin{construction}[Set interpretation of the groupoid-syntax \href{https://csl26-cohtt.github.io/TT.Groupoid.SetInterp.html}{\Agda}]\label{con:grp2set}
  We define the following functions mutually by pattern matching on
  the groupoid-syntax where $\Set :\equiv (X:\Type)\times((e\,e':x_0=_X x_1)\to e=e')$.
  \begin{alignat*}{10}
    & \llbracket\blank\rrbracket && : \Con_\G\to\Set                                                                                            && \llbracket\blank\rrbracket && : \Ty_\G\,\Gamma\to\llbracket\Gamma\rrbracket_{.1}\to\Set \\                                          
    & \llbracket\blank\rrbracket && : \Sub_\G\,\Delta\,\Gamma\to\llbracket\Delta\rrbracket_{.1}\to\llbracket\Gamma\rrbracket_{.1} \hspace{4.5em} &&\llbracket\blank\rrbracket && : \Tm_\G\,\Gamma\,A\to(\gamma:\llbracket\Gamma\rrbracket_{.1})\to(\llbracket A\rrbracket\,\gamma)_{.1}
  \end{alignat*}
  The cases for the constructors are analogous to the ones in Problem
  \ref{prob:type_interpretation}, with additional proofs of
  truncation-preservation: e.g.\ the empty context needs that $\top$
  is a set, context extension needs that $\Sigma$ preserves
  set-truncation. $\U$ is interpreted by $X$, $\El$ by $Y$. We
  interpret the extra truncation constructors as follows: we prove
  $\mathsf{isGrpdTy}$ by the fact that $\Set$ forms a groupoid, while
  functions between sets are sets, which proves $\mathsf{isSetSub}$
  and $\mathsf{isSetTm}$. All 1-dimensional equalities and the
  2-equalities $\U[\id]$, $\El[\id]$, $\Pi[\id]$ are interpreted by
  $\refl$, while the 2-equalities $\U[\circ]$, $\El[\circ]$,
  $\Pi[\circ]$ use cubical filling because these include compositions
  in the formalisation (this could be avoided using the technique
  explained in Section \ref{sec:meta}).
\end{construction}
The groupoid-syntax can be trivially interpreted into the set-syntax:
\begin{construction}[Set-syntax interpretation of the groupoid-syntax \href{https://csl26-cohtt.github.io/TT.Groupoid.ToSet.html}{\Agda}]\label{con:grp2setsyn}
 By pattern matching: 
  \begin{alignat*}{10}
    & \llbracket\blank\rrbracket && : \Con_\G\to\Con_\S                                                                                      && \llbracket\blank\rrbracket && : \Ty_\G\,\Gamma\to\Ty_\S\,\llbracket\Gamma\rrbracket \\                        
    & \llbracket\blank\rrbracket && : \Sub_\G\,\Delta\,\Gamma\to\Sub_\S\,\llbracket\Delta\rrbracket\,\llbracket\Gamma\rrbracket \hspace{5em} && \llbracket\blank\rrbracket && : \Tm_\S\,\Gamma\,A\to\Tm_\S\,\llbracket\Gamma\rrbracket\,\llbracket A\rrbracket
  \end{alignat*}
  Everything is interpreted by the corresponding component in the
  set-syntax, except (i) $\mathsf{isGrpdTy}_\G$ is interpreted by
  applying cumulativity of truncation levels to $\mathsf{isSetTy}_\S$;
  (ii) the higher equalities $\U[{\circ}],\dots,\Pi[\id]$ are
  interpreted by $\mathsf{isSetTy}_\S$.
\end{construction}

\section{\texorpdfstring{$\alpha$-normalisation for the groupoid-syntax}{α-normalisation for the groupoid-syntax}}\label{sec:alpha}

In this section we prove that although elements of $\Ty_\G$ in the
groupoid-syntax are only groupoid-truncated, they form a set. We define the
set of $\alpha$-normal forms for $\Ty_\G$, and then we show that every
$\Ty_\G$ is a retract of its $\alpha$-normal
forms. $\alpha$-normalisation is the process of eliminating explicit
instantiations from types along the substitution laws for types.

\subsection{\texorpdfstring{$\alpha$-normal forms}{α-normal forms}}

\begin{defn}[$\alpha$-normal forms \href{https://csl26-cohtt.github.io/TT.Groupoid.NTy.html\#NTy}{\Agda}]
  $\alpha$-normal forms are given by the inductive family $\NTy$ which
  is defined mutually with the quote function
  $\ulcorner\blank\urcorner$. We overload constructor names and
  metavariables, but use \textcolor{BrickRed}{brick red colour} for disambiguation.
  \begin{alignat*}{10}
    & \NTy && : \Con_\G\to\Type                     && \ulcorner\blank\urcorner && : \NTy\,\Gamma\to\Ty_\G\,\Gamma \\                                         
    & \N\U && : \NTy\,\Gamma                       && \ulcorner\N\U\urcorner && :\equiv \U_\G \\                                                             
    & \N\El && : \Tm_\G\,\Gamma\,\U_\G\to\NTy\,\Gamma  && \ulcorner\N\El\,\hat{A}\urcorner && :\equiv \El_\G\,\hat{A} \\                                         
    & \N\Pi && : (\N A:\NTy\,\Gamma)\to\NTy\,(\Gamma\ext_\G\ulcorner \N A\urcorner)\to\NTy\,\Gamma \hspace{3em} && \ulcorner\N\Pi\,\N A\,\N B\urcorner && :\equiv \Pi_\G\,\ulcorner \N A\urcorner\,\ulcorner \N B\urcorner
  \end{alignat*}
\end{defn}
It is not obvious that $\alpha$-normal forms are a set because $\NTy$
is indexed by $\Con_\G$ which contains elements of $\Ty_\G$ for which at
this point we do not know that it forms a set. $\NTy$ also includes
non-normal terms (via $\N\El$), hence we cannot rely on decidability
of equality and Hedberg's theorem \cite[Section
  7.2]{HoTTbook}. However, we can still show the following.
\begin{lemm}[\href{https://csl26-cohtt.github.io/TT.Groupoid.NTy.html\#isSetNTy}{\Agda}]
  $\NTy\,\Gamma$ forms a set.
\end{lemm}
\begin{proof}
We use the encode-decode method \cite{HoTTbook} to characterise
equality of $\NTy$. The cover (or code) relation is defined by
double-recursion on $\NTy$, mutually with the decode function.
\begin{alignat*}{10}
  & \rlap{$\Cover : \NTy\,\Gamma\to\NTy\,\Gamma\to\Type \hspace{10.5em}\decode : \Cover\,\N{A_0}\,\N{A_1}\to \N{A_0} = \N{A_1}$} \\
  & \Cover\,\N\U\,&& \N\U && :\equiv \top \\
  & \Cover\,(\N\El\,\hat{A_0})\,&& (\N\El\,\hat{A_1}) && :\equiv \hat{A_0} = \hat{A_1} \\
  & \Cover\,(\N\Pi\,\N{A_0}\,\N{B_0})\,&&(\N\Pi\,\N{A_1}\,\N{B_1}) && :\equiv (\N{A_2} : \Cover\,\N{A_0}\,\N{A_1})\times\Cover\,((\decode\,\N{A_2})_*\,\N{B_0})\,\N{B_1} \\
  & \Cover\,\_ && \_ && :\equiv \bot
\end{alignat*}
The $\decode$ function is defined by double-induction on $\N{A_0}$ and
$\N{A_1}$. Again, by double induction on $\NTy$, we prove that $\Cover$
is a proposition. By mutual induction on $\NTy$, we prove that
$\Cover$ is reflexive and decoding this reflexivity proof gives an
identity (reflexivity) path.
\begin{alignat*}{10}
  & \mathsf{reflCode} && : (\N A : \NTy\,\Gamma)\to\Cover\,\N A\,\N A && \hspace{2.1em}
  & \mathsf{decRefl} && : (\N A : \NTy\,\Gamma)\to \decode\,(\mathsf{reflCode}\,\N A) = \refl
\end{alignat*}
We use these and J to define encode and prove that $\decode$ is a
retraction:
\begin{alignat*}{10}
  & \mathsf{encode} && : \N{A_0} = \N{A_1} \to \Cover\,\N{A_0}\,\N{A_1} && \hspace{1.9em}
  & \mathsf{decEnc} && : (\N{A_2} : \N{A_0} = \N{A_1})\to \decode\,(\mathsf{encode}\,\N{A_2}) = \N{A_2}
\end{alignat*}
As retractions preserve homotopy levels, from
$\Cover\,\N{A_0}\,\N{A_1}$ being a proposition, we obtain that
$\N{A_0} = \N{A_1}$ is a proposition, hence $\NTy\,\Gamma$ is a set.
\end{proof}

\subsection{\texorpdfstring{$\alpha$-normalisation}{α-normalisation}}

We want to show that $\ulcorner\blank\urcorner :
\NTy\,\Gamma\to\Ty_\G\,\Gamma$ is a retraction, which will imply that
$\Ty_\G\,\Gamma$ is a set. For this, we define the other direction which
is the normalisation function and its completeness.
\begin{notation}
  For the rest of this section, as we only talk about the
  groupoid-syntax, we do not write the $_\G$ subscripts, so $\Ty$ means
  $\Ty_\G$, $\U$ means $\U_\G$, and so on.
\end{notation}
\begin{problem}[$\alpha$-normalisation \href{https://csl26-cohtt.github.io/TT.Groupoid.NTy.html\#norm}{\Agda}]\label{prob:norm}
  We define the following two functions by mutual induction on the
  groupoid-syntax.
\begin{alignat*}{10}
  & \norm && : \Ty\,\Gamma\to\NTy\,\Gamma && \hspace{3em}
  & \comp && : (A:\Ty\,\Gamma)\to\ulcorner\norm\,A\urcorner = A
\end{alignat*}
\end{problem}
\begin{proof}[Construction]
  On $\U$ and $\El$, the construction is trivial.
  \begin{alignat*}{10}
    & \norm\,\U && :\equiv \N\U  \hspace{3em} && \norm\,(\El\,\hat{A}) && :\equiv \N\El\,\hat{A} && \hspace{3em}
    & \comp\,\U && :\equiv \refl \hspace{3em} && \comp\,(\El\,\hat{A}) && :\equiv \refl
  \end{alignat*}
  On $\Pi$, we normalise recursively, but as $\norm\,B :
  \NTy\,(\Gamma\ext A)$, we need to transport it over completeness of
  $A$ to obtain an $\NTy\,(\Gamma\ext\ulcorner\norm\,A\urcorner)$:
  \begin{alignat*}{10}
    & \norm\,(\Pi\,A\,B) :\equiv \N\Pi\,(\norm\,A)\,\big((\comp\,A)_*\,\norm\,B\big) \\
    & \comp\,(\Pi\,A\,B) : \ulcorner\norm\,(\Pi\,A\,B)\urcorner \equiv \\
    & \hphantom{\comp\,(\Pi\,A\,B) : {}} \Pi\,\ulcorner\norm\,A\urcorner\,\ulcorner(\comp\,A)_*\,(\norm\,B)\urcorner \overset{\comp\,A}{=}
    \Pi\,A\,\ulcorner\norm\,B\urcorner \overset{\comp\,B}{=}
    \Pi\,A\,B
  \end{alignat*}
  To define $\norm$ on instantiated types, we need to instantiate
  normal forms. For this, we first show the following.
  \begin{problem}[\href{https://csl26-cohtt.github.io/TT.Groupoid.NTy.html\#_\%5B_\%5D\%E1\%B5\%80\%E1\%B4\%BA}{\Agda}]\label{prob:inst}
    $\NTy$ can be equipped with an instantiation operation
    $\blank\N[\blank\N]$ which is functorial, and
    $\ulcorner\blank\urcorner$ is a ``2-natural transformation''\footnote{
    We do not formally show that $\Ty$ and $\NTy$ form 2-functors, we use the phrase for intuition.} into
    $\Ty$, as follows (note the difference in colours for the
    overloaded names).
    \begin{alignat*}{10}
      & \blank\N[\blank\N] : \NTy\,\Gamma\to\Sub\,\Delta\,\Gamma\to\NTy\,\Delta \hspace{1.4em}
      \N[\circ\N] : \forall \N A\,\gamma\,\delta\ldotp\N A\N[\gamma\circ\delta\N] = \N A\N[\gamma\N]\N[\delta\N] \hspace{1.4em}
      \N[\id\N] : \forall\N A\ldotp\N A\N[\id\N] = \N A \\
      & \ulcorner\urcorner[] : \forall\N A\,\gamma\ldotp\ulcorner \N A\urcorner[\gamma] = \ulcorner \N A\N[\gamma\N]\urcorner \\
      & \ulcorner\urcorner[\circ] : \forall{\N A}\,\gamma\,\delta\ldotp
      \begin{tikzcd}[ampersand replacement=\&]
	{\ulcorner\N A\urcorner[\gamma\circ\delta]} \& {\ulcorner\N A\urcorner[\gamma][\delta]} \\
	\& {\ulcorner\N A\N[\gamma\N]\urcorner[\delta]} \\
	{\ulcorner\N A\N[\gamma\circ\delta\N]\urcorner} \& {\ulcorner\N A\N[\gamma\N]\N[\delta\N]\urcorner}
	\arrow["{[\circ]\,\ulcorner\N A\urcorner\,\gamma\,\delta}", from=1-1, to=1-2]
	\arrow["{\ulcorner\urcorner[]\,\N A\,(\gamma\circ\delta)}"', from=1-1, to=3-1]
	\arrow["{\ulcorner\urcorner[]\,\N A\,\gamma}", from=1-2, to=2-2]
	\arrow["{\ulcorner\urcorner[]\,(\N A\N[\gamma\N])\,\delta}", from=2-2, to=3-2]
	\arrow["{\N[\circ\N]\,\N A\,\gamma\,\delta}"', from=3-1, to=3-2]
      \end{tikzcd}
      \hspace{1.2em}
      \ulcorner\urcorner[\id] : \forall{\N A}\ldotp
      \begin{tikzcd}[ampersand replacement=\&]
        {\ulcorner\N A\urcorner[\id]} \\
        {\ulcorner\N A\N[\id\N]\urcorner} \& {\ulcorner\N A\urcorner}
        \arrow["{\ulcorner\urcorner[]\,\N A\,\id}"', from=1-1, to=2-1]
        \arrow["{[\id]\,\ulcorner\N A\urcorner}", from=1-1, to=2-2]
        \arrow["{\N[\id\N]\,\N A}"', from=2-1, to=2-2]
      \end{tikzcd}
    \end{alignat*}
  \end{problem}
  \begin{proof}[Construction for Problem \ref{prob:inst}]
    Instantiation of normal types is by mutual induction with
    naturality of $\ulcorner\blank\urcorner$. Instantiating $\N\U$
    just changes the implicit context arguments, instantiating $\N\El$
    means instantiating the term (which is an ordinary $\Tm_\G$ term,
    and is not normal), instantiating $\N\Pi$ is recursive:
    \begin{alignat*}{10}
      & \N\U\N[\gamma\N] && :\equiv \N\U \hspace{4em}
      & (\N\El\,\hat{A})\N[\gamma\N] && :\equiv \N\El\,(\hat{A}[\gamma]^\U) \hspace{4em}
      & (\N\Pi\,\N A\,\N B)\N[\gamma\N] && :\equiv \Pi\,(\N A\N[\gamma\N])\,(\N B\N[\gamma^{\ulcorner+\urcorner}\N])
    \end{alignat*}
    The operation $\blank\N[\blank^{\ulcorner+\urcorner}\N]$ used in
    the codomain of $\Pi$ is defined as follows. It also comes with a
    filler equation.
    \begin{alignat*}{10}
      & \blank\N[\blank^{\ulcorner+\urcorner}\N] && : \NTy\,(\Gamma\ext\ulcorner\N A\urcorner)\to(\gamma:\Sub\,\Delta\,\Gamma)\to\NTy\,(\Delta\ext \ulcorner\N A\N[\gamma\N]\urcorner) \\
      & \N B\N[\gamma^{\ulcorner+\urcorner}\N] && :\equiv (\ulcorner\urcorner[]\,\N A\,\gamma)_*\,(\N B\N[\gamma^+\N]) \\
      & \rlap{$\N B\N[\gamma^{\ulcorner+\urcorner}\N]\mathsf{filler} : \N B\N[\gamma^+\N] =^{\ulcorner\urcorner[]\,\N A\,\gamma} \N B\N[\gamma^{\ulcorner+\urcorner}\N]$}
    \end{alignat*}
    Analogously to $[{\circ}^+]$ and $[\id^+]$ of Definition
    \ref{def:groupoid}, we define their ``normal substitution''
    versions $\N[{\circ}^+\N]$ and $\N[\id^+\N]$.
    Naturality is reusing the substitution law of the corresponding
    syntactic operation, and in the case of $\N\Pi\,\N A\,\N B$,
    naturality for $\N A$ and $\N B$ are used (in the codomain of
    $\Pi$, both $\blank\N[\blank^{\ulcorner+\urcorner}\N]$ and its
    filler are used):
    \begin{alignat*}{10}
      & \ulcorner\urcorner[]\,\N\U\,\gamma && {}: \ulcorner\N\U\urcorner[\gamma] \equiv \U[\gamma] \overset{\U[]\,\gamma}{=} \U \equiv \ulcorner\N\U\N[\gamma\N]\urcorner \\
      & \ulcorner\urcorner[]\,(\N\El\,\hat{A})\,\gamma && {}: \ulcorner\N\El\,\hat{A}\urcorner[\gamma] \equiv (\El\,\hat{A})[\gamma] \overset{\El[]\,\hat{A}\,\gamma}{=} \El\,(\hat{A}[\gamma]^\U) \equiv \ulcorner\N(\El\,\hat{A})\N[\gamma\N]\urcorner \\
      & \ulcorner\urcorner[]\,(\N\Pi\,\N A\,\N B)\,\gamma && {}:
      \ulcorner\N\Pi\,\N A\,\N B\urcorner[\gamma] \equiv
      (\Pi\,\ulcorner\N A\urcorner\,\ulcorner\N B\urcorner)[\gamma] \overset{\Pi[]\,\ulcorner\N A\urcorner\,\ulcorner\N B\urcorner\,\gamma}{=}
      \Pi\,(\ulcorner\N A\urcorner[\gamma])\,(\ulcorner\N B\urcorner[\gamma^+]) \overset{\ulcorner\urcorner[]\,\N B\,\gamma^+}{=} \\
      & \rlap{$\hspace{5em}\Pi\,(\ulcorner\N A\urcorner[\gamma])\,\ulcorner\N B\N[\gamma^+\N]\urcorner \overset{\Pi\,(\ulcorner\urcorner[]\,\N A\,\gamma)\,\ulcorner\N B\N[\gamma^{\ulcorner+\urcorner}\N]\mathsf{filler}\urcorner}{=}
      \Pi\,\ulcorner\N A\N[\gamma\N]\urcorner\,\ulcorner\N B\N[\gamma^{\ulcorner+\urcorner}\N]\urcorner \equiv
      \ulcorner(\N\Pi\,\N A\,\N B)\N[\gamma\N]\urcorner$}
    \end{alignat*}
    The functoriality equation $\N[{\circ}\N]$ and the 2-naturality
    square $\ulcorner\urcorner[\circ]$ are proven mutually by
    induction on $\NTy$. The composition functor law for $\N\U$ is
    definitional, for $\N\El$ it reuses the functor law for terms of
    type $\U$, for $\N\Pi$ it is recursive:
    \begin{alignat*}{10}
      & \N[{\circ}\N]\,\N\U\,\gamma\,\delta && {}: \N\U\N[\gamma\circ\delta\N] \equiv \N\U \equiv \N\U\N[\gamma\N]\N[\delta\N] \\
      & \N[{\circ}\N]\,(\N\El\,\hat{A})\,\gamma\,\delta && {}: (\N\El\,\hat{A})\N[\gamma\circ\delta\N] \equiv
      \N\El\,(\hat{A}[\gamma\circ\delta]^\U) \overset{[\circ]^\U\,\hat{A}\,\gamma\,\delta}{=}
      \N\El\,(\hat{A}[\gamma]^\U[\delta]^\U) \equiv
      (\N\El\,\hat{A})\N[\gamma\N]\N[\delta\N] \\
      & \N[{\circ}\N]\,(\N\Pi\,\N A\,\N B)\,\gamma\,\delta && {}:
      (\N\Pi\,\N A\,\N B)\N[\gamma\circ\delta\N] \equiv
      \N\Pi\,(\N A\N[\gamma\circ\delta\N])\,(\N B\N[(\gamma\circ\delta)^{\ulcorner+\urcorner}\N]) \overset{\N\Pi\,(\N[{\circ}\N]\,\N A\,\gamma\,\delta)\,(\N[{\circ}^{\ulcorner+\urcorner}\N]\,\N B\,\gamma\,\delta)}{=} \\
      & && \hphantom{{}:{}} \N\Pi\,(\N A\N[\gamma\N]\N[\delta\N])\,(\N B\N[\gamma^{\ulcorner+\urcorner}\N]\N[\delta^{\ulcorner+\urcorner}\N]) \equiv
      (\N\Pi\,\N A\,\N B)\N[\gamma\N]\N[\delta\N]
    \end{alignat*}
    In the codomain part of the proof for $\N\Pi$ above, we used
    functoriality of the $\blank\N[\blank^{\ulcorner+\urcorner}\N]$
    operation which is defined by the dotted line (given by composition) in the following
    left square which is over the right square. We also give name to
    the filler of the left square.
    \[
    \begin{tikzcd}[ampersand replacement=\&]
	{\N B\N[(\gamma\circ\delta)^+\N]} \& {\N B\N[\gamma^+\N]\N[\delta^+\N]} \&\& {\ulcorner\N A\urcorner[\gamma\circ\delta]} \& {\ulcorner\N A\urcorner[\gamma][\delta]} \\
	\& {\hspace{2em}\N B\N[\gamma^{\ulcorner+\urcorner}\N]\N[\delta^+\N]} \&\&\& {\ulcorner\N A\N[\gamma\N]\urcorner[\delta]} \\
	{\N B\N[(\gamma\circ\delta)^{\ulcorner+\urcorner}\N]} \& {\N B\N[\gamma^{\ulcorner+\urcorner}\N]\N[\delta^{\ulcorner+\urcorner}\N]} \&\& {\ulcorner\N A\N[\gamma\circ\delta\N]\urcorner} \& {\ulcorner\N A\N[\gamma\N]\N[\delta\N]\urcorner}
	\arrow[""{name=0, anchor=center, inner sep=0}, "{\N[\circ^+\N]\,\N B\,\gamma\,\delta}", from=1-1, to=1-2]
	\arrow["{\N B\N[(\gamma\circ\delta)^{\ulcorner+\urcorner}\N]\mathsf{filler}}"', from=1-1, to=3-1]
	\arrow["{\N B\N[\gamma^{\ulcorner+\urcorner}\N]\mathsf{filler}}"', from=1-2, to=2-2]
	\arrow["{[\circ]\,\ulcorner\N A\urcorner\,\gamma\,\delta}", from=1-4, to=1-5]
	\arrow["{\ulcorner\urcorner[]\,\N A\,(\gamma\circ\delta)}"', from=1-4, to=3-4]
	\arrow["{\ulcorner\urcorner[\circ]\,\N A\,\gamma\,\delta}"{description}, draw=none, from=1-4, to=3-5]
	\arrow["{\ulcorner\urcorner[]\,\N A\,\gamma}", from=1-5, to=2-5]
	\arrow["{(\N B\N[\gamma^{\ulcorner+\urcorner}\N])\N[\delta^{\ulcorner+\urcorner}\N]\mathsf{filler}}"', from=2-2, to=3-2]
	\arrow["{\ulcorner\urcorner[]\,(\N A\N[\gamma\N])\,\delta}", from=2-5, to=3-5]
	\arrow[""{name=1, anchor=center, inner sep=0}, "{\N[\circ^{\ulcorner+\urcorner}\N]\,\N B\,\gamma\,\delta}"', dotted, from=3-1, to=3-2]
	\arrow["{\N[\circ\N]\,\N A\,\gamma\,\delta}"', from=3-4, to=3-5]
	\arrow["{\N[\circ^{\ulcorner+\urcorner}\N]\mathsf{filler}\,\N B\,\gamma\,\delta}"{description}, draw=none, from=0, to=1]
    \end{tikzcd}
    \]
    The $\ulcorner\urcorner[\circ]$-squares for $\N\U$ and $\N\El$ are
    definitionally the same as $\U[\circ]$ and $\El[\circ]$,
    respectively. We present the diagrammatic proof of $\U[\circ]$ for
    clarity, where double line means definitional equality. In this
    diagram, the inner and outer squares are definitionally
    equal. The square for $\N\Pi$ is more involved, we present it in Figure \ref{fig:picomp} in the Appendix.
    \begin{alignat*}{10}
      & {\ulcorner\urcorner[\circ]\,\N\U\,\gamma\,\delta :\equiv \U[\circ]\,\gamma\,\delta}
      \hspace{4em}{
      \begin{tikzcd}[ampersand replacement=\&,column sep=small]
	{\ulcorner\N\U\urcorner[\gamma\circ\delta]} \&\&\& {\ulcorner\N\U\urcorner[\gamma][\delta]} \\
	\& {\U[\gamma\circ\delta]} \& {\U[\gamma][\delta]} \\
	\&\& {\U[\delta]} \& {\ulcorner\N\U\N[\gamma\N]\urcorner[\delta]} \\
	\& \U \& \U \\
	{\ulcorner\N\U\N[\gamma\circ\delta\N]\urcorner} \&\&\& {\ulcorner\N\U\N[\gamma\N]\N[\delta\N]\urcorner}
	\arrow[""{name=0, anchor=center, inner sep=0}, "{[\circ]\,\ulcorner\N\U\urcorner\,\gamma\,\delta}", from=1-1, to=1-4]
	\arrow[equals, from=1-1, to=2-2]
	\arrow[""{name=1, anchor=center, inner sep=0}, "{\ulcorner\urcorner[]\,\N\U\,(\gamma\circ\delta)}"', from=1-1, to=5-1]
	\arrow[""{name=2, anchor=center, inner sep=0}, "{\ulcorner\urcorner[]\,\N\U\,\gamma}", from=1-4, to=3-4]
	\arrow[""{name=3, anchor=center, inner sep=0}, "{\raisebox{0.5em}{$[\circ]\,\U\,\gamma\,\delta$}}", from=2-2, to=2-3]
	\arrow[""{name=4, anchor=center, inner sep=0}, "{\U[]\,(\gamma\circ\delta)}"', from=2-2, to=4-2]
	\arrow["{\U[\circ]\,\gamma\,\delta}"{description}, draw=none, from=2-2, to=4-3]
	\arrow[equals, from=2-3, to=1-4]
	\arrow[""{name=5, anchor=center, inner sep=0}, "{\U[]\,\gamma}", from=2-3, to=3-3]
	\arrow[equals, from=3-3, to=3-4]
	\arrow[""{name=6, anchor=center, inner sep=0}, "{\U[]\,\delta}", from=3-3, to=4-3]
	\arrow[""{name=7, anchor=center, inner sep=0}, "{\ulcorner\urcorner[]\,(\N\U\N[\gamma\N])\,\delta}", from=3-4, to=5-4]
	\arrow[""{name=8, anchor=center, inner sep=0}, equals, from=4-2, to=4-3]
	\arrow[equals, from=4-3, to=5-4]
	\arrow[equals, from=5-1, to=4-2]
	\arrow[""{name=9, anchor=center, inner sep=0}, "{\N[\circ\N]\,\N\U\,\gamma\,\delta}"', from=5-1, to=5-4]
	\arrow[shorten <=8pt, shorten >=18pt, equals, from=0, to=3]
	\arrow[shorten <=12pt, shorten >=38pt, equals, from=1, to=4]
	\arrow[shorten <=24pt, shorten >=24pt, equals, from=2, to=5]
	\arrow[shorten <=24pt, shorten >=24pt, equals, from=7, to=6]
	\arrow[shorten <=12pt, shorten >=12pt, equals, from=9, to=8]
        \end{tikzcd}}
        \\
      & \ulcorner\urcorner[\circ]\,(\N\El\,\hat{A})\,\gamma\,\delta :\equiv \El[\circ]\,\hat{A}\,\gamma\,\delta \text{, see also Figure \ref{fig:El}} \\
      & \ulcorner\urcorner[\circ]\,(\N\Pi\,\N A\,\N B)\,\gamma\,\delta : \text{see Figure \ref{fig:picomp}}
    \end{alignat*}
    The functoriality equation $\N[\id\N]$ is proven by mutual
    induction on $\NTy$.
    \begin{alignat*}{10}
      & \N[\id\N]\,\N\U : \N\U\N[\id\N] \equiv \N\U \hspace{12.4em}
      \N[\id\N]\,(\N\El\,\hat{A}) : (\N\El\,\hat{A})\N[\id\N] \equiv \N\El\,(\hat{A}[\id]^\U) \overset{[\id]^\U\,\hat{A}}{=} \N\El\,\hat{A} \\
      & \N[\id\N]\,(\N\Pi\,\N A\,\N B) : (\N\Pi\,\N A\,\N B)\N[\id\N] \equiv
      \N\Pi\,(\N A\N[\id\N])\,(\N B\N[\id^{\ulcorner+\urcorner}\N]) \overset{\N\Pi\,(\N[\id\N]\,\N A)\,(\N[\id^{\ulcorner+\urcorner}\N]\,\N B)}{=}
      \N\Pi\,\N A\,\N B
    \end{alignat*}
    In the codomain part of the proof for $\N\Pi$ above, we used
    functoriality of the $\blank\N[\blank^{\ulcorner+\urcorner}\N]$
    operation which is defined by the dotted line in the following
    upper triangle which is over the lower triangle. We also give a name
    to the filler of the upper triangle.
    \[
    \begin{tikzcd}[ampersand replacement=\&,column sep=huge]
	{\N B\N[\id^+\N]} \&\& {\ulcorner\N A\urcorner[\id]} \\
	{\N B\N[\id^{\ulcorner+\urcorner}\N]} \& {\N B} \& {\ulcorner\N A\N[\id\N]\urcorner} \& {\ulcorner\N A\urcorner}
	\arrow["{\N B\N[\id^{\ulcorner+\urcorner}\N]\mathsf{filler}}"', from=1-1, to=2-1]
	\arrow[""{name=0, anchor=center, inner sep=0}, "{\N[\id^+\N]\,\N B}", curve={height=-18pt}, from=1-1, to=2-2]
	\arrow["{\ulcorner\urcorner[]\,\N A\,\id}"', from=1-3, to=2-3]
	\arrow[""{name=1, anchor=center, inner sep=0}, "{[\id]\,\ulcorner\N A\urcorner}", curve={height=-18pt}, from=1-3, to=2-4]
	\arrow["{\N[\id^{\ulcorner+\urcorner}\N]\,\N B}"', dotted, from=2-1, to=2-2]
	\arrow["{\N[\id\N]\,\N A}"', from=2-3, to=2-4]
	\arrow["{\N[\id^{\ulcorner+\urcorner}\N]\mathsf{filler}\,\N B}"{description}, draw=none, from=2-1, to=0]
	\arrow["{\ulcorner\urcorner[\id]\,\N A}"{description}, draw=none, from=2-3, to=1]
    \end{tikzcd}
    \]
    The 2-naturality triangle $\ulcorner\urcorner[\id]$ is proven by
    induction on $\NTy$ as follows:
    \begin{alignat*}{10}
      & \ulcorner\urcorner[\id]\,\N\U && :\equiv \U[\id] \hspace{5em}
      & \ulcorner\urcorner[\id]\,(\N\El\,\hat{A}) && :\equiv \El[\id]\,\hat{A} \hspace{5em}
      & \ulcorner\urcorner[\id]\,(\N\Pi\,\N A\,\N B) && : \text{see Figure \ref{fig:piid}}
    \end{alignat*}
This finishes the construction for Problem \ref{prob:inst}.
\end{proof}
  So far, we defined $\norm$ and $\comp$ on $\U$, $\El$ and $\Pi$. On
  substituted types, we define normalisation and its completeness as follows.
  \begin{alignat*}{10}
    & \norm\,(A[\gamma]) :\equiv (\norm\,A)\N[\gamma\N] \\
    & \comp\,(A[\gamma]) : \ulcorner\norm\,(A[\gamma])\urcorner \equiv
    \ulcorner(\norm\,A)\N[\gamma\N]\urcorner \overset{\ulcorner\urcorner[]\,(\norm\,A)\,\gamma}{=}
    \ulcorner\norm\,A\urcorner[\gamma] \overset{\comp\,A}{=}
    A[\gamma]
  \end{alignat*}
  The action of $\norm$ on the functor laws is the corresponding
  functor law for instantiation of normal types, i.e. 
  $\norm\,([{\circ}]\,A\,\gamma\,\delta) :\equiv \N[{\circ}\N]\,(\norm\,A)\,\gamma\,\delta$ and
  $\norm\,([\id]\,A) :\equiv \N[\id\N]\,(\norm\,A)$.
  Completeness for the functor laws is the filling of the following
  two squares: \\
  \adjustbox{minipage={\textwidth},scale=0.9,left}{
  \[\begin{tikzcd}[column sep=large]
	{\ulcorner(\norm\,A)\N[\gamma\circ\delta\N]\urcorner} & {\ulcorner(\norm\,A)\N[\gamma\N]\N[\delta\N]\urcorner} \\
	& {\ulcorner(\norm\,A)\N[\gamma\N]\urcorner[\delta]} \\
	{\ulcorner\norm\,A\urcorner[\gamma\circ\delta]} & {\ulcorner(\norm\,A)\urcorner[\gamma][\delta]} \\
	{A[\gamma\circ\delta]} & {A[\gamma][\delta]}
	\arrow[""{name=0, anchor=center, inner sep=0}, "{\N[\circ\N]\,(\norm\,A)\,\gamma\,\delta}", from=1-1, to=1-2]
	\arrow["{\ulcorner\urcorner[]\,(\norm\,A)\,(\gamma\circ\delta)}"', from=1-1, to=3-1]
	\arrow["{\ulcorner\urcorner\,((\norm\,A)\N[\gamma\N])\,\delta}", from=1-2, to=2-2]
	\arrow["{\ulcorner\urcorner\,(\norm\,A)\,\gamma}", from=2-2, to=3-2]
	\arrow[""{name=1, anchor=center, inner sep=0}, "{[\circ]\,\ulcorner\norm\,A\urcorner\,\gamma\,\delta}", from=3-1, to=3-2]
	\arrow["{\comp\,A}"', from=3-1, to=4-1]
	\arrow["{\comp\,A}", from=3-2, to=4-2]
	\arrow[""{name=2, anchor=center, inner sep=0}, "{[\circ]\,A\,\gamma\,\delta}"', from=4-1, to=4-2]
	\arrow["{\ulcorner\urcorner[\circ]\,(\norm\,A)\,\gamma\,\delta}"{description, pos=0.3}, draw=none, from=0, to=1]
	\arrow["{\mathrm{nat}}"{description}, draw=none, from=1, to=2]
  \end{tikzcd}
  \begin{tikzcd}[ampersand replacement=\&,column sep=huge]
	{\ulcorner(\norm\,A)\N[\id\N]\urcorner} \& {\ulcorner\norm\,A\urcorner} \\
	{\ulcorner\norm\,A\urcorner[\id]} \& {\ulcorner\norm\,A\urcorner} \\
	{A[\id]} \& A
	\arrow["{\ulcorner\N[\id\N]\,(\norm\,A)\urcorner}", from=1-1, to=1-2]
	\arrow["{\ulcorner\urcorner[]\,(\norm\,A)\,\id}"', from=1-1, to=2-1]
	\arrow["{\ulcorner\urcorner[\id]\,(\norm\,A)}"{description}, draw=none, from=1-1, to=2-2]
	\arrow[equals, from=1-2, to=2-2]
	\arrow["{[\id]\,\ulcorner\norm\,A\urcorner}"', from=2-1, to=2-2]
	\arrow["{\comp\,A}"', from=2-1, to=3-1]
	\arrow["{\comp\,A}", from=2-2, to=3-2]
	\arrow["{\mathrm{nat}}"{description}, draw=none, from=3-1, to=2-2]
	\arrow["{[\id]\,A}"', from=3-1, to=3-2]
  \end{tikzcd}
  \]
  }
  \vspace{0.1em}
  
  \noindent The action of $\norm$ on the substitution laws for $\U$ and $\El$ is
  given by $\refl$, and $\comp$ is given by trivial fillers for
  degenerate squares. The actions of $\norm$ and $\comp$ on
  $\Pi[]\,A\,B\,\gamma$ only involve naturality squares and fillers,
  they are presented in Figures \ref{fig:normpisub} and
  \ref{fig:comppisub} in the appendix. 

  The rest of the $\Ty$-paths that $\norm$ and $\comp$ have to preserve are the
  2-paths $\U[\circ]$, $\U[\id]$, $\El[\circ]$, $\El[\id]$,
  $\Pi[\circ]$, $\Pi[\id]$. As $\norm$ returns in a set, these are all
  trivial. The function $\comp$ produces an equality between elements
  of $\Ty$, and as $\Ty$ is a groupoid, it trivially preserves
  2-paths.
  Having defined $\norm$ and $\comp$, we finished the construction for Problem \ref{prob:norm}.
\end{proof}

\begin{thm}[\href{https://csl26-cohtt.github.io/TT.Groupoid.NTy.html\#isSetTy}{\Agda}]\label{thm:main}
  $\Ty_\G\,\Gamma$ is a set.
\end{thm}
\begin{proof}
Together, $\norm$ and $\comp$ witness that $\ulcorner\blank\urcorner$ is a
  retraction, which preserves h-levels: as $\NTy\,\Gamma$ is a set, so is
  $\Ty\,\Gamma$.
\end{proof}
\begin{remark}
Stability of normalisation is also provable, but we do not need it in this
paper.
\end{remark}

\section{Reaping the fruits}\label{sec:fruits}

\begin{problem}[\href{https://csl26-cohtt.github.io/TT.Groupoid.IsoSet.html}{\Agda}]\label{prob:iso}
  The set-syntax is isomorphic to the groupoid-syntax.
\end{problem}
\begin{proof}[Construction]
  In Construction \ref{con:grp2setsyn}, we defined the map from the
  groupoid-syntax to the set-syntax. Now we define the opposite
  direction using that $\Ty_\G\,\Gamma$ is a set. The roundtrips are proven by
  two simple inductions.
\end{proof}

\begin{construction}[Set interpretation of the set-syntax
  \href{https://csl26-cohtt.github.io/TT.Set.SetInterp.html}{\Agda}]
    We compose the groupoid-interpretation of the set syntax (Problem
  \ref{prob:iso}) and the set interpretation of the groupoid-syntax
  (Construction \ref{con:grp2set}).
\end{construction}

Groupoid CwFs are essentially algebras of the substitution calculus
part of the groupoid-syntax (Definition \ref{def:groupoid}), but we
also include three coherence laws for types (the pentagon law
$[\mathsf{ass}]$ and two identity triangles).
\begin{defn}[Groupoid CwF, GCwF \href{https://csl26-cohtt.github.io/TT.Groupoid.CwF.html}{\Agda}]\label{def:groupidCwF}
  An Ehrhard-style groupoid CwF is a 1-category (objects named $\Con :
  \Type$, morphisms $\Sub : \Con\to\Con\to\Set$), a 2-presheaf of
  types (given by $\Ty : \Con\to\mathsf{Groupoid}$, $\blank[\blank] :
  \Ty\,\Gamma\to\Sub\,\Delta\,\Gamma\to\Ty\,\Delta$, $[\circ] :
  A[\gamma\circ\delta] = A[\gamma][\delta]$, $[\id] : A[\id] = A$,
  $[\mathsf{ass}]$, $[\mathsf{idl}]$, $[\mathsf{idr}]$ as depicted
  below), a dependent presheaf of terms over types ($\Tm :
  (\Gamma:\Con)\to\Ty\,\Gamma\to\Set$, with instantiation and functor
  laws), with Ehrhard-style comprehension (operations
  $\blank\ext\blank$, $\blank^+$, $\p$, $\q$, $\langle\blank\rangle$
  with 8 equations as in Definition \ref{def:wild}).
\begin{alignat*}{10}
  & [\mathsf{ass}] : \forall A\,\gamma\,\delta\,\theta\ldotp && [\mathsf{idl}] : \forall A\,\gamma\ldotp && [\mathsf{idr}] : \forall A\,\gamma\ldotp \\
  &
  \begin{tikzcd}[ampersand replacement=\&]
	{A[\gamma\circ(\delta\circ\theta)]} \& {A[(\gamma\circ\delta)\circ\theta]} \\
	\& {A[\gamma\circ\delta][\theta]} \\
	{A[\gamma][\delta\circ\theta]} \& {A[\gamma][\delta][\theta]}
	\arrow["{\mathsf{ass}\,\gamma\,\delta\,\theta}", from=1-1, to=1-2]
	\arrow["{[\circ]\,A\,\gamma\,(\delta\circ\theta)}"', from=1-1, to=3-1]
	\arrow["{[\circ]\,A\,(\gamma\circ\delta)\,\theta}"', from=1-2, to=2-2]
	\arrow["{[\circ]\,A\,\gamma\,\delta}"', from=2-2, to=3-2]
	\arrow["{[\circ]\,(A[\gamma])\,\delta\,\theta}"', from=3-1, to=3-2]
  \end{tikzcd}
  &&
  \begin{tikzcd}[ampersand replacement=\&]
	{A[\id\circ\gamma]} \\
	{A[\id][\gamma]} \& {A[\gamma]} \\
	{}
	\arrow["{[\circ]\,A\,\id\,\gamma}"', from=1-1, to=2-1]
	\arrow["{\mathsf{idl}\,\gamma}", from=1-1, to=2-2]
	\arrow["{[\id]\,A}"', from=2-1, to=2-2]
  \end{tikzcd}
  &&
  \begin{tikzcd}[ampersand replacement=\&]
	{A[\gamma\circ\id]} \\
	{A[\gamma][\id]} \& {A[\gamma]} \\
	{}
	\arrow["{[\circ]\,A\,\gamma\,\id}"', from=1-1, to=2-1]
	\arrow["{\mathsf{idr}\,\gamma}", from=1-1, to=2-2]
	\arrow["{[\id]\,(A[\gamma])}"', from=2-1, to=2-2]
  \end{tikzcd}
\end{alignat*}
\end{defn}
\begin{remark}[\href{https://csl26-cohtt.github.io/TT.Groupoid.CwF.html\#\%5B\%5D\%E1\%B5\%80-idr}{\Agda}]\label{remark}
  In any groupoid CwF, $[\mathsf{idl}]$ and $[\mathsf{idr}]$ are interderivable.
  The direction $[\mathsf{idl}] \to [\mathsf{idr}]$ is described in
  Figure \ref{fig:idlidr} in the appendix. The same proof in the
  context of monoidal categories appears in \cite[Theorem 7]{KELLY1964397}.
\end{remark}
\begin{prop}[\href{https://csl26-cohtt.github.io/TT.Groupoid.NTy.html\#Coh}{\Agda}]
  In the groupoid-syntax (Definition \ref{def:groupoid}), the laws $[\mathsf{ass}]$,
  $[\mathsf{idl}]$ and $[\mathsf{idr}]$ are admissible.
\end{prop}
\begin{proof}
  Direct consequence of Theorem \ref{thm:main}.
\end{proof}

\section{Conclusions}\label{sec:conclusion}

We have presented a basic coherence theorem for \textbf{GCwF},
enabling the interpretation of the usual decidable intrinsic syntax of
type theory within models based on categories where the objects do not
form a set, such as the set model. Notably, we have achieved this
without relying on normalisation for the groupoid syntax or invoking
Hedberg's theorem. Furthermore, our method is adaptable, in principle,
to type theories without decidable equality. An interesting feature of
our approach is that it eliminates the need to explicitly incorporate
the usual coherence laws for 2-categories (such as the pentagon law)
into the syntax; these laws are admissible in our groupoid-syntax.

Despite these advancements, several significant challenges remain. For
instance, we aim to extend this framework to include a univalent
universe of propositions (i.e. $\mathsf{Prop}$ with propositional
extensionality). We also seek to address univalence for types without
introducing an additional universe, thereby demonstrating that
univalence can be soundly supported in this setting.

The addition of universes, even a minimal one such as a universe of
Booleans with large eliminations, would require a shift in our
methodology and might necessitate term normalisation. Extending the
framework to accommodate multiple universes would inevitably demand a
move to higher dimensions, introducing further complexity.

Our groupoid-syntax can be seen as the GCwF with $\Pi$ freely
generated from a set and a family over it. We would like to support
more interesting generating data, i.e.\ generating data which can
refer to the GCwF structure while being defined.

Finally, we would like to revisit the longstanding problem of modeling
semi-simplicial types within this context. One potential direction is
to extend our current recursive treatment of substitution and
substitution-related coherence laws, using these as a foundation to
systematically derive higher coherence conditions.

\bibliography{p}

\appendix

\section{More diagrams}


\begin{figure*}[h!]
  \adjustbox{minipage={\textwidth},scale=0.865,left}{
  \[
  \begin{tikzcd}[ampersand replacement=\&,column sep=huge]
	{\ulcorner\N\El\,\hat{A}\urcorner[\gamma\circ\delta]} \&\&\& {\ulcorner\N\El\,\hat{A}\urcorner[\gamma][\delta]} \\
	\& {\El\,\hat{A}[\gamma\circ\delta]} \& {(\El\,\hat{A})[\gamma][\delta]} \\
	\&\& {(\El\,(\hat{A}[\gamma]^\U))[\delta]} \& {\ulcorner(\N\El\,\hat{A})\N[\gamma\N]\urcorner[\delta]} \\
	\& {\El\,(\hat{A}[\gamma\circ\delta]^\U)} \& {\El\,(\hat{A}[\gamma]^\U[\delta]^\U)} \\
	{\ulcorner(\N\El\,\hat{A})\N[\gamma\circ\delta\N]\urcorner} \&\&\& {\ulcorner(\N\El\,\hat{A})\N[\gamma\N]\N[\delta\N]\urcorner}
	\arrow[""{name=0, anchor=center, inner sep=0}, "{[\circ]\,\ulcorner\N\El\,\hat{A}\urcorner\,\gamma\,\delta}", from=1-1, to=1-4]
	\arrow[equals, from=1-1, to=2-2]
	\arrow[""{name=1, anchor=center, inner sep=0}, "{\ulcorner\urcorner[]\,(\N\El\,\hat{A})\,(\gamma\circ\delta)}"', from=1-1, to=5-1]
	\arrow[""{name=2, anchor=center, inner sep=0}, "{\ulcorner\urcorner[]\,(\N\El\,\hat{A})\,\gamma}", from=1-4, to=3-4]
	\arrow[""{name=3, anchor=center, inner sep=0}, "{\raisebox{0.5em}{$[\circ]\,(\El\,\hat{A})\,\gamma\,\delta$}}", from=2-2, to=2-3]
	\arrow[""{name=4, anchor=center, inner sep=0}, "{\El[]\,\hat{A}\,(\gamma\circ\delta)}"', from=2-2, to=4-2]
	\arrow[equals, from=2-3, to=1-4]
	\arrow[""{name=5, anchor=center, inner sep=0}, "{\El[]\,\hat{A}\,\gamma}", from=2-3, to=3-3]
	\arrow[equals, from=3-3, to=3-4]
	\arrow[""{name=6, anchor=center, inner sep=0}, "{\El[]\,(\hat{A}[\gamma]^\U)\,\delta}", from=3-3, to=4-3]
	\arrow[""{name=7, anchor=center, inner sep=0}, "{\ulcorner\urcorner[]\,((\N\El\,\hat{A})\N[\gamma\N])\,\delta}", from=3-4, to=5-4]
	\arrow[""{name=8, anchor=center, inner sep=0}, "{[{\circ}]^\U\,\hat{A}\,\gamma\,\delta}"', from=4-2, to=4-3]
	\arrow[equals, from=4-3, to=5-4]
	\arrow[equals, from=5-1, to=4-2]
	\arrow[""{name=9, anchor=center, inner sep=0}, "{\N[\circ\N]\,(\N\El\,\hat{A})\,\gamma\,\delta}"', from=5-1, to=5-4]
	\arrow[shorten <=10pt, shorten >=20pt, equals, from=0, to=3]
	\arrow[shorten <=40pt, shorten >=65pt, equals, from=1, to=4]
	\arrow[shorten <=50pt, shorten >=58pt, equals, from=2, to=5]
	\arrow["{\El[\circ]\,\hat{A}\,\gamma\,\delta}"{description}, draw=none, from=3, to=8]
	\arrow[shorten <=52pt, shorten >=56pt, equals, from=7, to=6]
	\arrow[shorten <=10pt, shorten >=20pt, equals, from=9, to=8]
  \end{tikzcd}
  \]
  }
  \caption{This diagram is the proof
    $\ulcorner\urcorner[\circ]\,(\N\El\,\hat{A})\,\gamma\,\delta$, which
    is the outer square. Double lines mean definitional equality. The
    boundaries of the outer square are definitionally equal to the
    boundaries of the inner square, which we fill by
    $\El[\circ]\,\hat{A}\,\gamma\,\delta$.}\label{fig:El}
\end{figure*}

\begin{figure*}
\adjustbox{minipage={\textwidth},scale=0.8,left}{
\[
\begin{tikzcd}[ampersand replacement=\&,column sep=small]
	{\ulcorner\N\Pi\,\N A\,\N B\urcorner[\gamma\circ\delta]} \& {\ulcorner\N\Pi\,\N A\,\N B\urcorner[\gamma][\delta]} \&\& {(\Pi\,\ulcorner\N A\urcorner\,\ulcorner\N B\urcorner)[\gamma\circ\delta]} \& {(\Pi\,\ulcorner\N A\urcorner\,\ulcorner\N B\urcorner)[\gamma][\delta]} \\
	\&\&\&\& {(\Pi\,(\ulcorner\N A\urcorner[\gamma])\,(\ulcorner\N B\urcorner[\gamma^+]))[\delta]} \\
	\&\&\& {\Pi\,(\ulcorner\N A\urcorner[\gamma\circ\delta])\,(\ulcorner\N B\urcorner[(\gamma\circ\delta)^+])} \& {(\Pi\,(\ulcorner\N A\urcorner[\gamma])\,\ulcorner\N B\N[\gamma^+\N]\urcorner)[\delta]} \\
	\& {\ulcorner(\N\Pi\,\N A\,\N B)\N[\gamma\N]\urcorner[\delta]} \& \equiv \&\& {(\Pi\,\ulcorner\N A\N[\gamma\N]\urcorner\,\ulcorner\N B\N[\gamma^{\ulcorner+\urcorner}\N]\urcorner)[\delta]} \\
	\&\&\& {\Pi\,(\ulcorner\N A\urcorner[\gamma\circ\delta])\,\ulcorner\N B\N[(\gamma\circ\delta)^+\N]\urcorner} \& {\Pi\,(\ulcorner\N A\N[\gamma\N]\urcorner[\delta])\,(\ulcorner\N B\N[\gamma^{\ulcorner+\urcorner}\N]\urcorner[\delta^+])} \\
	\&\&\&\& {\Pi\,(\ulcorner\N A\N[\gamma\N ]\urcorner[\delta])\,\ulcorner\N B\N[\gamma^{\ulcorner+\urcorner}\N]\N[\delta^+\N]\urcorner} \\
	{\ulcorner(\N\Pi\,\N A\,\N B)\N[\gamma\circ\delta\N]\urcorner} \& {\ulcorner(\N\Pi\,\N A\,\N B)\N[\gamma\N]\N[\delta\N]\urcorner} \&\& {\Pi\,\ulcorner\N A\N[\gamma\circ\delta\N]\urcorner\,\ulcorner\N B\N[(\gamma\circ\delta)^{\ulcorner+\urcorner}\N]\urcorner} \& {\Pi\,\ulcorner\N A\N[\gamma\N]\N[\delta\N]\urcorner\,\ulcorner\N B\N[\gamma^{\ulcorner+\urcorner}\N]\N[\delta^{\ulcorner+\urcorner}\N]\urcorner}
	\arrow["{\raisebox{0.5em}{$[\circ]\,\ulcorner\N\Pi\,\N A\,\N B\urcorner\,\gamma\,\delta$}}", from=1-1, to=1-2]
	\arrow["{\ulcorner\urcorner[]\,(\N\Pi\,\N A\,\N B)\,(\gamma\circ\delta)}"{pos=0.4}, from=1-1, to=7-1]
	\arrow["{\ulcorner\urcorner[]\,(\N\Pi\,\N A\,\N B)\,\gamma}"', from=1-2, to=4-2]
	\arrow["{[\circ]\,(\Pi\,\ulcorner\N A\urcorner\,\ulcorner\N B\urcorner)\,\gamma\,\delta}", from=1-4, to=1-5]
	\arrow["{\Pi[]\,\ulcorner\N A\urcorner\,\ulcorner\N B\urcorner\,(\gamma\circ\delta)}"', from=1-4, to=3-4]
	\arrow["{\Pi[]\,\ulcorner\N A\urcorner\,\ulcorner\N B\urcorner\,\gamma}"', from=1-5, to=2-5]
	\arrow["{\ulcorner\urcorner[]\,\N B\,\gamma^+}"', from=2-5, to=3-5]
	\arrow["{\ulcorner\urcorner[]\,\N B\,((\gamma\circ\delta)^+)}"', from=3-4, to=5-4]
	\arrow["{\Pi\,(\ulcorner\urcorner[]\,\N A\,\gamma)\,\ulcorner\N B\N[\gamma^{\ulcorner+\urcorner}\N]\mathsf{filler}\urcorner}"', from=3-5, to=4-5]
	\arrow["{\ulcorner\urcorner[]\,((\N\Pi\,\N A\,\N B)\N[\gamma\N])\,\delta}"', from=4-2, to=7-2]
	\arrow["{\Pi[]\,\ulcorner\N A\N[\gamma\N]\urcorner\,\ulcorner\N B\N[\gamma^{\ulcorner+\urcorner}\N]\urcorner\,\delta}"', from=4-5, to=5-5]
	\arrow["{\Pi\,(\ulcorner\urcorner[]\,\N A\,(\gamma\circ\delta))\,\ulcorner\N B\N[(\gamma\circ\delta)^{\ulcorner+\urcorner}\N]\mathsf{filler}\urcorner}"', from=5-4, to=7-4]
	\arrow["{\ulcorner\urcorner[]\,(\N B\N[\gamma^{\ulcorner+\urcorner}\N])\,\delta^+}"', from=5-5, to=6-5]
	\arrow["{\Pi\,(\ulcorner\urcorner[]\,(\N A\N[\gamma\N])\,\delta)\,\ulcorner\N B\N[\gamma^{\ulcorner+\urcorner}\N]\N[\delta^{\ulcorner+\urcorner}\N]\mathsf{filler}\urcorner}"', from=6-5, to=7-5]
	\arrow["{\raisebox{-1em}{$\N[\circ\N]\,(\N\Pi\,\N A\,\N B)\,\gamma\,\delta$}}"', from=7-1, to=7-2]
	\arrow["{\raisebox{-1em}{$\Pi\,\ulcorner\N[{\circ}\N]\,\N A\,\gamma\,\delta\urcorner\,\ulcorner\N[{\circ}^{\ulcorner+\urcorner}\N]\,\N B\,\gamma\,\delta\urcorner$}}"', from=7-4, to=7-5]
\end{tikzcd}
\]
}

\adjustbox{minipage={\textwidth},scale=0.73,left}{
\[
\begin{tikzcd}[ampersand replacement=\&, column sep=tiny]
	{(\Pi\,\ulcorner\N A\urcorner\,\ulcorner\N B\urcorner)[\gamma\circ\delta]} \&\&\& {(\Pi\,\ulcorner\N A\urcorner\,\ulcorner\N B\urcorner)[\gamma][\delta]} \\
	\&\&\& {(\Pi\,(\ulcorner\N A\urcorner[\gamma])\,(\ulcorner\N B\urcorner[\gamma^+]))[\delta]} \\
	{\Pi\,(\ulcorner\N A\urcorner[\gamma\circ\delta])\,(\ulcorner\N B\urcorner[(\gamma\circ\delta)^+])} \&\& {\Pi\,(\ulcorner\N A\urcorner[\gamma][\delta])\,(\ulcorner\N B\urcorner[\gamma^+][\delta^+])} \\
	\& {\Pi\,(\ulcorner\N A\urcorner[\gamma][\delta])\,(\ulcorner\N B\urcorner[\gamma^+\circ\delta^+])} \&\& {(\Pi\,(\ulcorner\N A\urcorner[\gamma])\,\ulcorner\N B\N[\gamma^+\N]\urcorner)[\delta]} \\
	\&\& {\Pi\,(\ulcorner\N A\urcorner[\gamma][\delta])\,(\ulcorner\N B\N[\gamma^+\N]\urcorner[\delta^+])} \& {(\Pi\,\ulcorner\N A\N[\gamma\N]\urcorner\,\ulcorner\N B\N[\gamma^{\ulcorner+\urcorner}\N]\urcorner)[\delta]} \\
	\& {\Pi\,(\ulcorner\N A\urcorner[\gamma][\delta])\,\ulcorner\N B\N[\gamma^+\circ\delta^+\N]\urcorner} \&\& {\Pi\,(\ulcorner\N A\N[\gamma\N]\urcorner[\delta])\,(\ulcorner\N B\N[\gamma^{\ulcorner+\urcorner}\N]\urcorner[\delta^+])} \\
	{\Pi\,(\ulcorner\N A\urcorner[\gamma\circ\delta])\,\ulcorner\N B\N[(\gamma\circ\delta)^+\N]\urcorner} \&\& {\Pi\,(\ulcorner\N A\urcorner[\gamma][\delta])\,\ulcorner\N B\N[\gamma^+\N]\N[\delta^+\N]\urcorner} \\
	\&\&\& {\Pi\,(\ulcorner\N A\N[\gamma\N]\urcorner[\delta])\,\ulcorner\N B\N[\gamma^{\ulcorner+\urcorner}\N]\N[\delta^+\N]\urcorner} \\
	{\Pi\,\ulcorner\N A\N[\gamma\circ\delta\N]\urcorner\,\ulcorner\N B\N[(\gamma\circ\delta)^{\ulcorner+\urcorner}\N]\urcorner} \&\&\& {\Pi\,\ulcorner\N A\N[\gamma\N]\N[\delta\N]\urcorner\,\ulcorner\N B\N[\gamma^{\ulcorner+\urcorner}\N]\N[\delta^{\ulcorner+\urcorner}\N]\urcorner}
	\arrow["{[\circ]\,(\Pi\,\ulcorner\N A\urcorner\,\ulcorner\N B\urcorner)\,\gamma\,\delta}", from=1-1, to=1-4]
	\arrow["{\Pi[]\,\ulcorner\N A\urcorner\,\ulcorner\N B\urcorner\,(\gamma\circ\delta)}"', from=1-1, to=3-1]
	\arrow["{\Pi[]\,\ulcorner\N A\urcorner\,\ulcorner\N B\urcorner\,\gamma}", from=1-4, to=2-4]
	\arrow["{\Pi[]\,(\ulcorner\N A\urcorner[\gamma])\,(\ulcorner\N B\urcorner[\gamma^+])\,\delta}"', from=2-4, to=3-3]
	\arrow["{\ulcorner\urcorner[]\,\N B\,\gamma^+}", from=2-4, to=4-4]
	\arrow["{\mathrm{nat}}"{description}, draw=none, from=2-4, to=5-3]
	\arrow["{\Pi[\circ]\,\ulcorner\N A\urcorner\,\ulcorner\N B\urcorner\,\gamma\,\delta}"{description}, draw=none, from=3-1, to=1-4]
	\arrow[""{name=0, anchor=center, inner sep=0}, "{\Pi\,([\circ]\,\ulcorner\N A\urcorner\,\gamma\,\delta)\,([\circ^+]\,\ulcorner\N B\urcorner\,\gamma\,\delta)}", from=3-1, to=3-3]
	\arrow["\begin{array}{c} \substack{\Pi\,([\circ]\,\ulcorner\N A\urcorner\,\gamma\,\delta)\\(\ulcorner\N B\urcorner[{\circ}^+\,\gamma\,\delta])} \end{array}"', from=3-1, to=4-2]
	\arrow[""{name=1, anchor=center, inner sep=0}, "{\ulcorner\urcorner[]\,\N B\,((\gamma\circ\delta)^+)}"', from=3-1, to=7-1]
	\arrow["{\ulcorner\urcorner[]\,\N B\,\gamma^+}"{description}, from=3-3, to=5-3]
	\arrow["{[\circ]\,\ulcorner\N B\urcorner\,\gamma^+\,\delta^+}"', from=4-2, to=3-3]
	\arrow[""{name=2, anchor=center, inner sep=0}, "{\ulcorner\urcorner[]\,\N B\,(\gamma^+\circ\delta^+)}"{description}, from=4-2, to=6-2]
	\arrow["{\Pi[]\,(\ulcorner\N A\urcorner[\gamma])\,\ulcorner\N B\N[\gamma^+\N]\urcorner\,\delta}"{description}, from=4-4, to=5-3]
	\arrow["\begin{array}{c} \substack{(\Pi\,(\ulcorner\urcorner[]\,\N A\,\gamma)\qquad\\\ulcorner\N B\N[\gamma^{\ulcorner+\urcorner}\N]\mathsf{filler}\urcorner)[\delta]} \end{array}", from=4-4, to=5-4]
	\arrow["{\mathrm{nat}}"{description}, draw=none, from=5-3, to=5-4]
	\arrow["{\Pi\,((\ulcorner\urcorner[]\,\N A\,\gamma)[\delta])\,(\ulcorner\N B\N[\gamma^{\ulcorner+\urcorner}\N]\mathsf{filler}\urcorner[\delta^+])}"{description}, from=5-3, to=6-4]
	\arrow["{\ulcorner\urcorner[]\,(\N B\N[\gamma^+\N])\,\delta^+}"{description}, from=5-3, to=7-3]
	\arrow["{\mathrm{nat}}"{description}, draw=none, from=5-3, to=8-4]
	\arrow["\begin{array}{c} \substack{\Pi[]\,\ulcorner\N A\N[\gamma\N]\urcorner\\\ulcorner\N B\N[\gamma^{\ulcorner+\urcorner}\N]\urcorner\,\delta} \end{array}"{pos=0.4}, from=5-4, to=6-4]
	\arrow["{\N[\circ\N]\,\N B\,\gamma^+\,\delta^+}", from=6-2, to=7-3]
	\arrow["{\ulcorner\urcorner[]\,(\N B\N[\gamma^{\ulcorner+\urcorner}\N])\,\delta^+}", from=6-4, to=8-4]
	\arrow["\begin{array}{c} \substack{\Pi\,([\circ]\,\ulcorner\N A\urcorner\,\gamma\,\delta)\\\ulcorner\N B\N[\circ^+\,\gamma\,\delta\N]\urcorner} \end{array}", from=7-1, to=6-2]
	\arrow[""{name=3, anchor=center, inner sep=0}, "{\Pi\,([\circ]\,\ulcorner\N A\urcorner\,\gamma\,\delta)\,\ulcorner\N[\circ^+\N]\,\N B\,\gamma\,\delta\urcorner}"', from=7-1, to=7-3]
	\arrow["{\Pi\,(\ulcorner\urcorner[]\,\N A\,(\gamma\circ\delta))\,\ulcorner\N B\N[(\gamma\circ\delta)^{\ulcorner+\urcorner}\N]\mathsf{filler}\urcorner}", from=7-1, to=9-1]
	\arrow["{\Pi\,(\ulcorner\urcorner[\circ]\,\N A\,\gamma\,\delta)\,\ulcorner\N[\circ^{\ulcorner+\urcorner}\N]\mathsf{filler}\,\N B\,\gamma\,\delta\urcorner}"{description}, draw=none, from=7-1, to=9-4]
	\arrow["{\Pi\,((\ulcorner\urcorner[]\,\N A\,\gamma)[\delta])\,\ulcorner\N B\N[\gamma^{\ulcorner+\urcorner}\N]\mathsf{filler}\N[\delta^+\N]\urcorner}"', from=7-3, to=8-4]
	\arrow["{\Pi\,(\ulcorner\urcorner[]\,(\N A\N[\gamma\N])\,\delta)\,\ulcorner\N B\N[\gamma^{\ulcorner+\urcorner}\N]\N[\delta^{\ulcorner+\urcorner}\N]\mathsf{filler}\urcorner}"', from=8-4, to=9-4]
	\arrow["{\Pi\,\ulcorner\N[\circ\N]\,\N A\,\gamma\,\delta\urcorner\,\ulcorner\N[\circ^{\ulcorner+\urcorner}\N]\,\N B\,\gamma\,\delta\urcorner}"', from=9-1, to=9-4]
	\arrow["{\Pi\,([\circ]\,\ulcorner\N A\urcorner\,\gamma\,\delta)\,([\circ^+]\mathsf{filler}\,\ulcorner\N B\urcorner\,\gamma\,\delta)}"{description}, draw=none, from=0, to=4-2]
	\arrow["{\mathrm{nat}}"{description}, draw=none, from=1, to=2]
	\arrow["{\ulcorner\urcorner[\circ]\,\N B\,\gamma\,\delta}"{description, pos=0.7}, draw=none, from=2, to=5-3]
	\arrow["{\Pi\,([\circ]\,\ulcorner\N A\urcorner\,\gamma\,\delta)\,\ulcorner\N[\circ^+\N]\mathsf{filler}\,\N B\,\gamma\,\delta\urcorner}"{description}, draw=none, from=6-2, to=3]
\end{tikzcd}
\]
}
  \caption{This diagram is the proof
    $\ulcorner\urcorner[\circ]\,(\N\Pi\,A\,B)\,\gamma\,\delta$. In the
    upper part, we compute the square to be filled: the left hand side
    square is definitionally equal to the right hand side one. Then,
    we fill the right hand side square in the lower diagram, where the
    boundary of the square is the same as the upper right hand side
    square.}\label{fig:picomp}
\end{figure*}

\begin{figure*}
  \[
\begin{tikzcd}[ampersand replacement=\&]
	{\ulcorner\N \Pi\,\N A\,\N B\urcorner[\id]} \&\&\& {(\Pi\,\ulcorner\N A\urcorner\,\ulcorner\N B\urcorner)[\id]} \\
	\&\&\& {\Pi\,(\ulcorner\N A\urcorner[\id])\,(\ulcorner\N B\urcorner[\id^+])} \\
	\& {\ulcorner\N \Pi\,\N A\,\N B\urcorner} \&\& {\Pi\,(\ulcorner\N A\urcorner[\id])\,\ulcorner\N B\N[\id^+\N]\urcorner} \& {\Pi\,\ulcorner\N A\urcorner\,\ulcorner\N B\urcorner} \\
	{\ulcorner\N \Pi\,\N A\,\N B\N[\id\N]\urcorner} \&\&\& {\Pi\,\ulcorner\N A\N[\id\N]\urcorner\,\ulcorner(\N B\N[\id^{\ulcorner+\urcorner}\N])\urcorner}
	\arrow[""{name=0, anchor=center, inner sep=0}, "{[\id]\,\ulcorner\N \Pi\,\N A\,\N B\urcorner}", from=1-1, to=3-2]
	\arrow["{\ulcorner\urcorner[]\,(\N \Pi\,\N A\,\N B)\,\id}"', from=1-1, to=4-1]
	\arrow["{\Pi[]\,\ulcorner\N A\urcorner\,\ulcorner\N B\urcorner\,\id}"', from=1-4, to=2-4]
	\arrow["{[\id]\,(\Pi\,\ulcorner\N A\urcorner\,\ulcorner\N B\urcorner)}", from=1-4, to=3-5]
	\arrow[""{name=1, anchor=center, inner sep=0}, "{\ulcorner\urcorner[]\,\N B\,(\id^+)}"', from=2-4, to=3-4]
	\arrow["{\Pi\,(\ulcorner\urcorner[]\,\N A\,\id)\,(\ulcorner\N B\N[\id^{\ulcorner+\urcorner}\N]\urcorner)}"', from=3-4, to=4-4]
	\arrow["{\N[\id\N]\,(\N \Pi\,\N A\,\N B)}"', from=4-1, to=3-2]
	\arrow["{\Pi\,\ulcorner\N[\id\N]\,\N A\urcorner\,\ulcorner\N[\id^{\ulcorner+\urcorner}\N]\,\N B\urcorner}"', from=4-4, to=3-5]
	\arrow["\equiv"{description}, draw=none, from=0, to=1]
\end{tikzcd}
\]

  \[\begin{tikzcd}[column sep=huge]
    {(\Pi\,\ulcorner\N A\urcorner\,\ulcorner\N B\urcorner)[\id]} \\
    \\
    {\Pi\,(\ulcorner\N A\urcorner[\id])\,(\ulcorner\N B\urcorner[\id^+])} && {\Pi\,\ulcorner\N A\urcorner\,\ulcorner\N B\urcorner} \\
    & {\Pi\,\ulcorner\N A\urcorner\,(\ulcorner\N B\urcorner[\id])} \\
    & {\Pi\,\ulcorner\N A\urcorner\,\ulcorner\N B\N[\id\N]\urcorner} \\
    {\Pi\,(\ulcorner\N A\urcorner[\id])\,\ulcorner\N B\N[\id^+\N]\urcorner} && {\Pi\,\ulcorner\N A\urcorner\,\ulcorner\N B\urcorner} \\
    \\
    {\Pi\,\ulcorner\N A\N[\id\N]\urcorner\,\ulcorner(\N B\N[\id^{\ulcorner+\urcorner}\N])\urcorner}
    \arrow["{\Pi[]\,\ulcorner\N A\urcorner\,\ulcorner\N B\urcorner\,\id}"', from=1-1, to=3-1]
    \arrow[""{name=0, anchor=center, inner sep=0}, "{[\id]\,(\Pi\,\ulcorner\N A\urcorner\,\ulcorner\N B\urcorner)}", from=1-1, to=3-3]
    \arrow[""{name=1, anchor=center, inner sep=0}, "{\Pi\,([\id]\,\ulcorner\N A\urcorner)\,([\id^+]\,\ulcorner\N B\urcorner)}"{description}, from=3-1, to=3-3]
    \arrow["\begin{array}{c} \substack{\Pi\,([\id]\,\ulcorner\N A\urcorner)\\(\ulcorner\N B\urcorner[\id^+])} \end{array}"', from=3-1, to=4-2]
    \arrow[""{name=2, anchor=center, inner sep=0}, "{\ulcorner\urcorner[]\,\N B\,(\id^+)}"', from=3-1, to=6-1]
    \arrow[""{name=3, anchor=center, inner sep=0}, equals, from=3-3, to=6-3]
    \arrow["{[\id]\,\ulcorner\N B\urcorner}"', from=4-2, to=3-3]
    \arrow[""{name=4, anchor=center, inner sep=0}, "{\ulcorner\urcorner[]\,\N B\,\id}"', from=4-2, to=5-2]
    \arrow["{\N[\id\N]\,\N B}", from=5-2, to=6-3]
    \arrow["\begin{array}{c} \substack{\Pi\,([\id]\,\ulcorner\N A\urcorner)\\\ulcorner\N B\N[\id^+\N]\urcorner} \end{array}", from=6-1, to=5-2]
    \arrow[""{name=5, anchor=center, inner sep=0}, "{\Pi\,([\id]\,\ulcorner\N A\urcorner)\,\ulcorner\N [\id^+\N]\,\N B\urcorner}"{description}, from=6-1, to=6-3]
    \arrow["{\Pi\,(\ulcorner\urcorner[]\,\N A\,\id)\,(\ulcorner\N B\N[\id^{\ulcorner+\urcorner}\N]\urcorner)}"', from=6-1, to=8-1]
    \arrow[""{name=6, anchor=center, inner sep=0}, "{\Pi\,\ulcorner\N[\id\N]\,\N A\urcorner\,\ulcorner\N[\id^{\ulcorner+\urcorner}\N]\,\N B\urcorner}"', from=8-1, to=6-3]
    \arrow["{\Pi[\id]\,\ulcorner\N A\urcorner\,\ulcorner\N B\urcorner}"{description}, draw=none, from=3-1, to=0]
    \arrow["{\Pi\,([\id]\,\ulcorner\N A\urcorner)\,([\id^+]\mathsf{filler}\,\ulcorner\N B\urcorner)}"{description}, draw=none, from=1, to=4-2]
    \arrow["{\mathrm{nat}}"{description}, draw=none, from=2, to=4]
    \arrow["{\ulcorner\urcorner[\id]\,\N B}"{description}, draw=none, from=4, to=3]
    \arrow["{\Pi\,([\id]\,\ulcorner\N A\urcorner)\,\ulcorner\N[\id^+\N]\mathsf{filler}\,\N B\urcorner}"{description}, draw=none, from=5-2, to=5]
    \arrow["{\Pi\,(\ulcorner\urcorner[\id]\,\N A)\,\ulcorner\N[\id^{\ulcorner+\urcorner}\N]\mathsf{filler}\,\N B\urcorner}"{description}, draw=none, from=6-1, to=6]
  \end{tikzcd}\]
  \caption{This diagram is the proof
    $\ulcorner\urcorner[\id]\,(\N\Pi\,A\,B)$. In the upper part, we
    compute the triangle to be filled: the left hand side triangle is
    definitionally equal to the right hand side one. Then, we fill the
    right hand side triangle in the lower diagram, where we duplicate
    the vertex $\Pi\,\ulcorner\N A\urcorner\,\ulcorner\N B\urcorner$
    for readability.}\label{fig:piid}
\end{figure*}

\begin{figure*}
  \[
\begin{tikzcd}[ampersand replacement=\&]
	{((\comp\,A)_*\,(\norm\,B))\N[\gamma^+\N]} \& {(\norm\,B)\N[\gamma^+\N]} \\
	{((\comp\,A)_*\,(\norm\,B))\N[\gamma^{\ulcorner+\urcorner}\N]} \& {(\comp\,(A[\gamma]))_*\,((\norm\,B)\N[\gamma^+\N])} \\
	{\ulcorner\norm\,A\urcorner[\gamma]} \& {A[\gamma]} \\
	{\ulcorner(\norm\,A)\N[\gamma\N]\urcorner} \& {\ulcorner(\norm\,A)\N[\gamma\N]\urcorner}
	\arrow["{\mathrm{transportFiller}}", from=1-1, to=1-2]
	\arrow[""{name=0, anchor=center, inner sep=0}, "{((\comp\,A)_*\,(\norm\,B))\N[\gamma^{\ulcorner+\urcorner}\N]\mathsf{filler}}"', from=1-1, to=2-1]
	\arrow[""{name=1, anchor=center, inner sep=0}, "{\mathrm{transportFiller}}", from=1-2, to=2-2]
	\arrow["e"', dotted, from=2-1, to=2-2]
	\arrow["{\comp\,A}", from=3-1, to=3-2]
	\arrow["{\ulcorner\urcorner[]\,(\norm\,A)\,\gamma}"', from=3-1, to=4-1]
	\arrow["{\mathrm{fillerOf}\,(\comp\,(A[\gamma]))}"{description}, draw=none, from=3-1, to=4-2]
	\arrow["{\comp\,(A[\gamma])}", from=3-2, to=4-2]
	\arrow[equals, from=4-1, to=4-2]
	\arrow["{\mathrm{fillerOf}\,e}"{description, pos=0.6}, draw=none, from=0, to=1]
\end{tikzcd}
\]
\caption{Normalisation on the substitution law for $\Pi$ acts as
  follows: $\norm\,(\Pi[]\,A\,B\,\gamma) :\equiv \N\Pi\,\refl\,e$
  where $e$ is defined in the upper square in this diagram. The upper
  square is a dependent square over the lower
  one.}\label{fig:normpisub}
\end{figure*}
\begin{figure*}  
\adjustbox{minipage={\textwidth},scale=0.835,left}{
  \[
\begin{tikzcd}[ampersand replacement=\&]
	{\Pi\,\ulcorner(\norm\,A)\N[\gamma\N]\urcorner\ \ulcorner(((\comp\,A)_*\,(\norm\,B))\N[\gamma^{\ulcorner+\urcorner}\N]\urcorner} \& {\Pi\ \ulcorner\norm\,A\N[\gamma\N]\urcorner\ \ulcorner(\comp\,(A[\gamma]))_* (\norm\,B\N[\gamma^+\N])\urcorner} \\
	{\Pi\ (\ulcorner\norm\,A\urcorner[\gamma])\ \ulcorner((\comp\,A)_*\,(\norm\,B))\N[\gamma^+\N]\urcorner} \& {\Pi\ (A[\gamma])\ \ulcorner(\norm\,B)\N[\gamma^+\N]\urcorner} \\
	{\Pi\,(\ulcorner\norm\,A\urcorner[\gamma])\,(\ulcorner(\comp\,A)_*\,(\norm\,B)\urcorner[\gamma^+])} \\
	{(\Pi\ \ulcorner\norm\,A\urcorner\ \ulcorner(\comp\,A)_*\,(\norm\,B)\urcorner)[\gamma]} \& {\Pi\ (A[\gamma])\ (\ulcorner\norm\,B\urcorner[\gamma^+])} \\
	{(\Pi\ A\ \ulcorner\norm\,B\urcorner)[\gamma]} \\
	{(\Pi\ A\ B)[\gamma]} \& {\Pi\ (A[\gamma])\ (B[\gamma^+])}
	\arrow["e", from=1-1, to=1-2]
	\arrow["\begin{array}{c} \substack{\Pi\,(\ulcorner\urcorner[]\,(\norm\,A)\,\gamma)\\\ulcorner((\comp\,A)_*\,(\norm\,B))\N[\gamma^{\ulcorner+\urcorner}\N]\mathsf{filler}\urcorner} \end{array}"', from=1-1, to=2-1]
	\arrow["{\Pi\,(\mathrm{fillerOf}\,(\comp\,(A[\gamma])))\,(\mathrm{fillerOf}\,e)}"{description}, draw=none, from=1-1, to=2-2]
	\arrow["{\Pi\,(\comp\,(A[\gamma]))\,\mathrm{transportFiller}}", from=1-2, to=2-2]
	\arrow["{\Pi\,(\comp\,A)\,\mathrm{transportFiller}}"', from=2-1, to=2-2]
	\arrow[""{name=0, anchor=center, inner sep=0}, "{\ulcorner\urcorner[]\,((\comp\,A)_*\,(\norm\,B))\,\gamma^+}"', from=2-1, to=3-1]
	\arrow[""{name=1, anchor=center, inner sep=0}, "{\ulcorner\urcorner[]\,(\norm\,B)\,\gamma^+}", from=2-2, to=4-2]
	\arrow["{\Pi[]\,\ulcorner\norm\,A\urcorner\,\ulcorner(\comp\,A)_*\,(\norm\,B)\urcorner\,\gamma}"', from=3-1, to=4-1]
	\arrow["{\Pi\,(\comp\,A)\,\mathrm{transportFiller}}", from=3-1, to=4-2]
	\arrow["{\mathrm{nat}}"{description}, draw=none, from=4-1, to=4-2]
	\arrow["{\Pi\,(\comp\,A)\,\mathrm{transportFiller}}"', from=4-1, to=5-1]
	\arrow[""{name=2, anchor=center, inner sep=0}, "{\comp\,B}", from=4-2, to=6-2]
	\arrow["{\Pi[]\,A\,\ulcorner\norm\,B\urcorner\,\gamma}"', from=5-1, to=4-2]
	\arrow[""{name=3, anchor=center, inner sep=0}, "{\comp\,B}"', from=5-1, to=6-1]
	\arrow["{\Pi[]\,A\,B\,\gamma}"', from=6-1, to=6-2]
	\arrow["{\mathrm{nat}}"{description}, draw=none, from=0, to=1]
	\arrow["{\mathrm{nat}}"{description}, draw=none, from=3, to=2]
\end{tikzcd}
\]
}
  \caption{This diagram is the proof
    $\comp\,(\Pi[]\,A\,B\,\gamma)$. The line $e$ is defined in Figure
    \ref{fig:normpisub}.}\label{fig:comppisub}
\end{figure*}

\begin{figure*}
\adjustbox{minipage={\textwidth},scale=0.86,left}{
  \[\begin{tikzcd}
    {A[\gamma\circ\id]} \\
    \\
    \\
    \\
    & {A[\gamma\circ\id][\id]} \\
    \\
    \\
    \\
    && {A[(\gamma\circ\id)\circ\id]} \\
    \\
    && {A[\gamma\circ(\id\circ\id)]} \\
    &&& {A[\gamma\circ\id]} \\
    && {A[\gamma][\id\circ\id]} \\
    & {A[\gamma][\id][\id]} &&&& {A[\gamma][\id]} \\
    {A[\gamma][\id]} &&&&&&& {A[\gamma]}
    \arrow[""{name=0, anchor=center, inner sep=0}, "{[\circ]\,A\,\gamma\,\id}"', from=1-1, to=15-1]
    \arrow[""{name=1, anchor=center, inner sep=0}, "{\mathsf{idr}\,\gamma}", curve={height=-120pt}, from=1-1, to=15-8]
    \arrow["{[\id]\,(A[\gamma\circ\id))}"{description}, from=5-2, to=1-1]
    \arrow[""{name=2, anchor=center, inner sep=0}, "{[\circ]\,A\,\gamma\,\id}"', from=5-2, to=14-2]
    \arrow[""{name=3, anchor=center, inner sep=0}, "{\mathsf{idr}\,\gamma}", curve={height=-60pt}, from=5-2, to=14-6]
    \arrow["{[\circ]\,A\,(\gamma\circ\id)\,\id}"{description}, from=9-3, to=5-2]
    \arrow[""{name=4, anchor=center, inner sep=0}, "{\mathsf{idr}\,\gamma}", from=9-3, to=12-4]
    \arrow["{\mathsf{isSetSub}}"'{pos=0.4}, draw=none, from=9-3, to=12-4]
    \arrow["{\mathsf{ass}\,\gamma\,\id\,\id}", from=11-3, to=9-3]
    \arrow["{\mathsf{idl}\,\id}"', from=11-3, to=12-4]
    \arrow["{[\circ]\,A\,\gamma\,(\id\circ\id)}"', from=11-3, to=13-3]
    \arrow["{[\circ]\,A\,\gamma\,\id}"{description}, from=12-4, to=14-6]
    \arrow["{\mathrm{nat}}"{description}, draw=none, from=13-3, to=12-4]
    \arrow["{[\circ]\,(A[\gamma])\,\id\,\id}"{description}, from=13-3, to=14-2]
    \arrow["{\mathsf{idl}\,\id}", from=13-3, to=14-6]
    \arrow[""{name=5, anchor=center, inner sep=0}, "{[\id]\,(A[\gamma])}"', from=14-2, to=14-6]
    \arrow["{[\id]\,(A[\gamma][\id])}"{description}, from=14-2, to=15-1]
    \arrow["{[\id]\,(A[\gamma])}"{description}, from=14-6, to=15-8]
    \arrow[""{name=6, anchor=center, inner sep=0}, "{[\id]\,(A[\gamma])}"', from=15-1, to=15-8]
    \arrow["{\mathrm{nat}}"{description}, draw=none, from=2, to=0]
    \arrow["{\mathrm{nat}}"{description}, draw=none, from=3, to=1]
    \arrow["{\mathrm{nat}}"{description}, draw=none, from=4, to=3]
    \arrow["{[\mathsf{ass}]\,A\,\gamma\,\id\,\id}"{description, pos=0.4}, draw=none, from=11-3, to=2]
    \arrow["{[\mathsf{idl}]\,(A[\gamma])\,\id}"{description}, draw=none, from=13-3, to=5]
    \arrow["{\mathrm{nat}}"{description}, draw=none, from=5, to=6]
  \end{tikzcd}\]
  }
  \caption{Proof that $[\mathsf{idl}]$ implies $[\mathsf{idr}]$ in any
    groupoid CwF (Definition \ref{def:groupidCwF}).}\label{fig:idlidr}
\end{figure*}

\end{document}

%% file: abbrevs.tex
\newcommand{\blank}{\mathord{\hspace{1pt}\text{--}\hspace{1pt}}} 
\NewCommandCopy{\oldGamma}{\Gamma}\renewcommand{\Gamma}{\mathit{\oldGamma}}
\NewCommandCopy{\oldDelta}{\Delta}\renewcommand{\Delta}{\mathit{\oldDelta}}
\NewCommandCopy{\oldTheta}{\Theta}\renewcommand{\Theta}{\mathit{\oldTheta}}
\renewcommand{\alpha}{\upalpha}
\renewcommand{\beta}{\upbeta}
\renewcommand{\eta}{\upeta}
\renewcommand{\epsilon}{\upvarepsilon}
\newcommand{\Type}{\mathsf{Type}}
\newcommand{\Set}{\mathsf{Set}}
\newcommand{\Con}{\mathsf{Con}}
\newcommand{\Sub}{\mathsf{Sub}}
\newcommand{\id}{\mathsf{id}}
\newcommand{\Ty}{\mathsf{Ty}}
\newcommand{\Tm}{\mathsf{Tm}}
\newcommand{\ext}{\triangleright}
\newcommand{\p}{\mathsf{p}}
\newcommand{\q}{\mathsf{q}}
\newcommand{\U}{\mathsf{U}}
\newcommand{\El}{\mathsf{El}}
\newcommand{\app}{\mathsf{app}}
\newcommand{\lam}{\mathsf{lam}}
\newcommand{\refl}{\mathsf{refl}}

\renewcommand{\S}{\mathsf{S}}
\newcommand{\G}{\mathsf{G}}
\newcommand{\NTy}{\mathsf{NTy}}
\newcommand{\N}[1]{\mathcolor{BrickRed}{#1}}
\newcommand{\Cover}{\mathsf{Cover}}
\newcommand{\decode}{\mathsf{decode}}
\newcommand{\norm}{\mathsf{norm}}
\newcommand{\comp}{\mathsf{compl}}
\newcommand{\inU}{\mathsf{inU}}
\newcommand{\inEl}{\mathsf{inEl}}